\documentclass{article}

\pdfoutput=1
\input{cliffordcs.sty}

\title{\bf Optimal Two-Qubit Circuits for Universal Fault-Tolerant
  Quantum Computation}
\author{Andrew N.\ Glaudell,$^{1,2,3,4,}$\thanks{Corresponding author: Glaudell\_Andrew@bah.com}~
  Neil J.\ Ross,$^{5}$ and Jacob M.\ Taylor$^{1,2}$\\[2pt]
  \small $^{1}$ Institute for Advanced Computer Studies and Joint
  Center for Quantum Information and Computer Science,\\
  \small University of Maryland, College Park, MD, USA\\  
  \small $^{2}$ Joint Quantum Institute, University of Maryland,
  College Park, MD, USA \\
  \small $^{3}$ Booz Allen Hamilton, Annapolis Junction, MD, USA\\
  \small $^{4}$ Department of Mathematical Sciences, George Mason University, Fairfax, VA, USA\\
  \small $^{5}$ Department of Mathematics and Statistics, Dalhousie
  University, Halifax, NS, Canada}
\date{}

\begin{document}

\maketitle

\begin{abstract}
  We study two-qubit circuits over the Clifford+$CS$ gate set, which
  consists of the Clifford gates together with the controlled-phase
  gate $CS=\diag(1,1,1,i)$. The Clifford+$CS$ gate set is universal
  for quantum computation and its elements can be implemented
  fault-tolerantly in most error-correcting schemes through magic
  state distillation. Since non-Clifford gates are typically more
  expensive to perform in a fault-tolerant manner, it is often
  desirable to construct circuits that use few $CS$ gates. In the
  present paper, we introduce an efficient and optimal synthesis
  algorithm for two-qubit Clifford+$CS$ operators.  Our algorithm
  inputs a Clifford+$CS$ operator $U$ and outputs a Clifford+$CS$
  circuit for $U$, which uses the least possible number of $CS$
  gates. Because the algorithm is deterministic, the circuit it
  associates to a Clifford+$CS$ operator can be viewed as a normal
  form for that operator. We give an explicit description of these
  normal forms and use this description to derive a worst-case lower
  bound of $5\log_2(\frac{1}{\epsilon})+O(1)$ on the number of $CS$
  gates required to $\epsilon$-approximate elements of $\su(4)$. Our
  work leverages a wide variety of mathematical tools that may find
  further applications in the study of fault-tolerant quantum
  circuits.
\end{abstract}

\section{Introduction}
\label{sec:intro}

In the context of fault-tolerant quantum computing, operations from
the Clifford group are relatively easy to perform and are therefore
considered inexpensive. In contrast, operations that do not belong to
the Clifford group are complicated to execute fault-tolerantly because
they require resource intensive distillation protocols
\cite{reichardt}. Since non-Clifford operations are necessary for
universal quantum computing, it has become standard to use the number
of non-Clifford gates in a circuit as a measure of its cost. This
fault-tolerant perspective on the cost of circuits has profoundly
impacted the field of quantum compiling and significant efforts have
been devoted to minimizing the number of non-Clifford operations in
circuits.

An important problem in quantum compiling is the problem of
\emph{exact synthesis}: given an operator $U$ known to be exactly
representable over some gate set $G$, find a circuit for $U$ over
$G$. An \emph{exact synthesis algorithm} is a constructive solution to
this problem. When the gate set $G$ is an extension of the Clifford
group, it is desirable that the exact synthesis algorithm for $G$ be
efficient and produce circuits that use as few non-Clifford gates as
possible.

In the past few years, methods from algebraic number theory have been
successfully applied to the exact synthesis problem associated to a
variety of single-qubit \cite{bbg2014,PhysRevA.88.012313,FGKM15,
  kmm-exact,KY2015, vsynth, RS16} and single-qutrit \cite{BCKZ15,
  2271, KBS2013,PhysRevA.98.032304} gate sets. In many cases, the
resulting exact synthesis algorithms efficiently produce circuits that
are \emph{optimal}, in the sense that they use the least possible
number of non-Clifford gates. These powerful exact synthesis methods
were central in the development of good unitary approximation methods,
which play a key role in the compilation of practical quantum programs
\cite{bbg2014,PhysRevA.88.012313, KBRY2015, kmm-approx,vsynth, RS16}.

Exact synthesis algorithms also exist for various instantiations of
the multi-qubit compiling problem, though each suffers shortcomings in
some respect. Optimal algorithms for two-qubit circuits over
continuous gate sets have been known for a number of years
\cite{PhysRevA.69.062321,PhysRevA.67.042313}. Unfortunately, such gate
sets are not well-suited for fault-tolerant quantum
computing. Multi-qubit exact synthesis algorithms for universal and
fault-tolerant gate sets were introduced more recently
\cite{restricted,m3,GS13,m4,CH2018,m2, MSdM2018,m1}. While the
algorithms of \cite{restricted,GS13} are far from optimal, the
algorithms of \cite{m3,m4,m2,m1} synthesize provably optimal circuits
by cleverly utilizing certain properties of fault-tolerant gate sets
containing the Clifford group. However, the runtimes of these optimal
synthesis algorithms are exponential in both qubit count and optimal
circuit length. Powerful heuristics were introduced in \cite{m1}
achieving polynomial scaling with optimal circuit
length. Unfortunately, even this improved heuristic algorithm takes
thousands of seconds to compute optimal two-qubit circuits of
practical size (40 non-Clifford operations) on modest hardware.

Not only are these multi-qubit exact synthesis algorithms impractical
in many cases, they also fail to shed much light on the
\emph{structure} of optimal circuits. In the single-qubit case,
intimate knowledge of this structure for certain gate sets was
devleoped by describing optimal circuits via regular expressions or,
equivalently, automata \cite{ma-remarks}. Such descriptions are of
theoretical interest, but also have practical consequences. In
particular, for certain single-qubit gate sets these decriptions
allowed researchers to derive a rigorous lower-bound on the number of
non-Clifford gates required to approximate typical elements of
$\su(2)$ \cite{Sel2015}. Analogous statements about approximations of
multi-qubit unitaries have eluded researchers thus far.

In the present paper, we introduce an efficient and optimal exact
synthesis algorithm for a two-qubit gate set that is appropriate for
universal and fault-tolerant quantum computing. We focus on two-qubit
circuits over the Clifford+$CS$ gate set, which consists of the
Clifford gates together with the non-Clifford controlled-phase gate
$CS=\diag(1,1,1,i)$. The $CS$ gate has received recent attention as an
alternative to the $T$-gate in methods for fault-tolerant quantum
computing \cite{beverland2020lower,haah2018codes} and due to its
natural implementation as an entangling operation in certain
superconducting qubit systems
\cite{cross2016scalable,garion2020synthesis,garion2020experimental,PhysRevA.93.060302}
whose fidelity is approaching that of single-qubit gates
\cite{PhysRevLett.125.120504,PhysRevLett.125.240503}. Our algorithm
produces an optimal circuit in a number of arithmetic operations
linear in the length of the optimal decomposition. This is unlike
existing multi-qubit synthesis methods. Moreover, because our
algorithm is deterministic, the circuit it associates to a
Clifford+$CS$ operator can be viewed as a normal form for that
operator. We give an explicit description of these normal forms in the
language of automata and use this description to derive a worst-case
lower bound of $5\log_2(\frac{1}{\epsilon})+O(1)$ on the number of
$CS$ gates required to $\epsilon$-approximate elements of $\su(4)$. A
Mathematica package implementing our algorithm is freely available
on-line \cite{thecode}. This code is very efficient, synthesizing
optimal circuits of $CS$-count 10000 in $1.2\pm0.1$ seconds on modest
hardware.

The paper is structured as follows. We first introduce a convenient
set of generators in \cref{sec:gens}. Then, in \cref{sec:iso}, we
describe the exceptional isomorphism
$\mbox{SU}(4)\cong\mbox{Spin}(6)$. In \cref{sec:synth}, we leverage
this isomorphism to introduce an exact synthesis algorithm for
Clifford+$CS$ operators. In \cref{sec:nfs}, we use the theory of
automata to study the structure of the circuits produced by the exact
synthesis algorithm. We take advantage of this structure in
\cref{sec:lowerbounds} to establish a worst-case lower bound on the
number of non-Clifford resources required to $\epsilon$-approximate
elements of $\su(4)$ using Clifford+$CS$ circuits. Finally, we
conclude and discuss avenues for future work in \cref{sec:conc}.

\section{Generators}
\label{sec:gens}

Throughout, we use $\N$, $\Z$, $\R$, and $\C$ to denote the usual
collection of numbers, $\Z_p$ to denote the collection integers modulo
$p$, and $\Zi$ to denote the collection of Gaussian integers (the
complex numbers with integer real and imaginary parts). We write
$\rho$ for the canonical homomorphism $\Z \to \Z_2$ (if $n\in\Z$ then
$\rho(n)$ is the parity of $n$). For two integers $n\leq m$, we write
$[n,m]$ for the set $\s{n,\ldots,m}\subseteq \Z$ and simply write
$[m]$ for $[1,m]$. We view scalars and vectors as matrices so that any
concept defined for matrices of arbitrary dimensions also applies to
scalars and vectors. Finally, for readability, we use the symbol
$\cdot$ to denote the zero entries of a matrix.

The single-qubit \emph{Pauli} gates $X$, $Y$, and $Z$ are defined as
\[
X= \begin{bmatrix} \cdot & 1 \\ 1 & \cdot \end{bmatrix},
\qquad
Y= \begin{bmatrix} \cdot & -i \\ i & \cdot \end{bmatrix},
\qquad \mbox{ and } \qquad
Z = \begin{bmatrix} 1 & \cdot \\ \cdot & -1 \end{bmatrix}.
\]
These gates generate the \emph{single-qubit Pauli group} $\s{i^a P ~;~
  a\in\Z_4 \mbox{ and } P\in\s{I, X, Y, Z}}$. The \emph{two-qubit
Pauli group}, which we denote by $\pauli$, is defined as $\pauli =
\s{i^a (P\otimes Q) ~;~ a\in\Z_4 \mbox{ and } P,Q \in
  \s{I,X,Y,Z}}$. The \emph{Clifford} gates $H$, $S$, and $CZ$ are
defined as
\[
H= \frac{1}{\sqrt{2}}\begin{bmatrix} 1 & 1 \\ 1 & -1 \end{bmatrix},
\quad
S= \begin{bmatrix} 1 & \cdot \\ \cdot & i \end{bmatrix},
\qquad \mbox{ and } \qquad
CZ =
\begin{bmatrix}
1 & \cdot & \cdot & \cdot \\
\cdot & 1 & \cdot & \cdot \\
\cdot & \cdot & 1 & \cdot \\
\cdot & \cdot & \cdot & -1
\end{bmatrix}.
\]
These gates are known as the \emph{Hadamard} gate, the \emph{phase}
gate, and the \emph{controlled-$Z$} gate, respectively. The
\emph{single-qubit Clifford group} is generated by $H$ and $S$ and
contains the primitive 8-th root of unity $\omega =
e^{\frac{i\pi}{4}}$. The \emph{two-qubit Clifford group}, which we
denote by $\clifford$, consists of the operators which can be
represented by a two-qubit circuit over the gate set $\s{H, S,
  CZ}$. Equivalently, $\clifford$ is generated by $H\otimes I$,
$I\otimes H$, $S \otimes I$, $I\otimes S$, and $CZ$. Up to global
phases, the Clifford groups are the normalizers of the Pauli groups.

Clifford gates are well-suited for fault-tolerant quantum computation
but the Clifford group is not universal. One can obtain a universal
group by extending $\clifford$ with the \emph{controlled-phase gate}
$CS$ defined as
\[
CS = \begin{bmatrix} 
	1 & \cdot & \cdot & \cdot \\
	\cdot & 1 & \cdot & \cdot \\
	\cdot & \cdot & 1 & \cdot \\
	\cdot & \cdot & \cdot & i \\  
\end{bmatrix}.
\]
In what follows, we focus on the group $\cliffordcs$ of operators
which can be represented by a two-qubit circuit over the universal
gate set $\s{H, S, CZ, CS}$. Equivalently, $\cliffordcs$ is the group
generated by $H\otimes I$, $I\otimes H$, $S \otimes I$, $I\otimes S$,
$CZ$, and $CS$. We have $\pauli\subseteq\clifford \subseteq
\cliffordcs$.  We sometimes refer to $\cliffordcs$ as the
\emph{Clifford+$CS$} group or \emph{Clifford+controlled-phase}
group. We know from \cite{restricted} that $\cliffordcs$ is the group
of $4\times4$ unitary matrices of the form
\begin{equation}
  \label{eq:u4rep}
  \frac{1}{\sqrt{2}^k} M
\end{equation}
where $k\in\N$ and the entries of $M$ belong to $\Zi$. In the
fault-tolerant setting, the $CS$ gate is considered vastly more
expensive than any of the Clifford gates. As a result, the cost of a
Clifford+$CS$ circuit is determined by its \emph{$CS$-count}: the
number of $CS$ gates that appear in the circuit. Our goal is to find
circuits for the elements of $\cliffordcs$ that are optimal in
$CS$-count.

We start by introducing a generalization of the $CS$ gate which will
be helpful in describing the elements of $\cliffordcs$.

\begin{definition}
  \label{def:gens}
  Let $P$ and $Q$ be distinct elements of $\pauli \setminus \s{\Id}$
  such that $P$ and $Q$ are Hermitian and $PQ=QP$. Then $R(P,Q)$ is
  defined as
  \[
  R(P,Q) = \exp \left( \frac{i\pi}{2} \left( \frac{\Id-P}{2} \right)
  \left( \frac{\Id-Q}{2} \right) \right).
  \]  
\end{definition}

We have $R(Z\otimes \Id, \Id\otimes Z)=CS$. Moreover, since
$\clifford$ normalizes $\pauli$ and $CR(P,Q)C^\dagger = R(CPC^\dagger,
CQC^\dagger)$ for every $C\in \clifford$, we know that
$R(P,Q)\in\cliffordcs$ for every appropriate $P,Q\in\pauli$. We record
some important properties of the $R(P,Q)$ gates in the lemma
below. Because the proof of the lemma is tedious but relatively
straightforward, it is given in \cref{app:proof}.

\begin{lemma}
  \label{lem:rels}
  Let $C\in\clifford$ and let $P$, $Q$, and $L$ be distinct elements
  of $\pauli \setminus \s{I}$. Assume that $P$, $Q$, and $L$ are
  Hermitian and that $PQ=QP$, $PL=LP$, and $QL=-LQ$. Then the
  following relations hold:
  \begin{align}
  C R(P,Q)C^\dagger & = R(CPC^\dagger,CQC^\dagger), \label{eq:CliffordCommute}\\    
  R(P,Q) & = R(Q,P), \label{eq:swappable}\\  
  R(P,-PQ) & = R(P,Q), \label{eq:permutable}\\
  R(P,-Q) & \in R(P,Q) \clifford, \label{eq:minusPauli}\\
  R(P,Q)^2 & \in \clifford,\mbox{ and} \label{eq:squared}\\
  R(P,L) R(P,Q) & = R(P,Q) R(P,iQL). \label{eq:sharedPauli}
  \end{align}
\end{lemma}

We will use the $R(P,Q)$ gates of \cref{def:gens} to define normal
forms for the elements of $\cliffordcs$. The equivalences given by
\cref{lem:rels} show that it is not necessary to use every $R(P,Q)$
gate and the following definition specifies the ones we will be using.

\begin{definition}
  \label{def:genset}
  Let $\mathcal{T}_1$ and $\mathcal{T}_2$ be the subsets of
  $\pauli\times \pauli$ given below.
  \begin{align*}
  & \mathcal{T}_1 = \s{(P,Q) ~;~ P\in\s{X\otimes I, Y\otimes I, Z\otimes
        I}, Q \in \s{I\otimes X, I\otimes Y, I\otimes Z}} \\
  & \mathcal{T}_2 = \s{(P,Q) ~;~ P\in\s{X\otimes X, Z\otimes X, Y\otimes
        X}, Q \in \s{Y\otimes Y, Z\otimes Y, X\otimes Y}, \mbox{ and }
      PQ=QP}.
  \end{align*}
  The set $\gens$ is defined as $\gens = \s{R(P,Q) ~;~ (P,Q) \in
    \mathcal{T}_1 \mbox{ or } (P,Q) \in \mathcal{T}_2}$.
\end{definition}

\begin{figure}
  \centering
  \begin{tabular}{lllll}
    $R(X\otimes I,I\otimes X)$ & $R(Y\otimes I,I\otimes Y)$ &
    $R(Z\otimes I,I\otimes Z)$ & $R(Y\otimes I,I\otimes Z)$ &
    $R(Z\otimes I,I\otimes Y)$ \\
    $R(Z\otimes I,I\otimes X)$ & $R(X\otimes I,I\otimes Z)$ &
    $R(X\otimes I,I\otimes Y)$ & $R(Y\otimes I,I\otimes X)$ &
    $R(X\otimes X,Y\otimes Y)$ \\
    $R(X\otimes X,Z\otimes Y)$ & $R(Z\otimes X,Y\otimes Y)$ &
    $R(Y\otimes X,X\otimes Y)$ & $R(Z\otimes X,X\otimes Y)$ &
    $R(Y\otimes X, Z\otimes Y)$
  \end{tabular}
  \caption{The 15 elements of $\gens$.}\label{fig:elems}
\end{figure}

The set $\gens$ contains 15 elements which are explicitly listed in
\cref{fig:elems}. It can be verified that all of the elements of
$\gens$ are distinct, even up to right-multiplication by a Clifford
gate. It will be helpful to consider the set $\gens$ ordered as in
\cref{fig:elems}, which is to be read left-to-right and row-by-row. We
then write $\gens_j$ to refer to the $j$-th element of $\gens$. For
example, $\gens_1$ is in the top left of \cref{fig:elems}, $\gens_5$
is in the top right, and $\gens_{15}$ is in the bottom right. The
position of $R(P,Q)$ in this ordering roughly expresses the complexity
of the Clifford circuit required to conjugate $CS$ to $R(P,Q)$.

We close this section by showing that every element of $\cliffordcs$
can be expressed as a sequence of elements of $\gens$ followed by a
single element of $\clifford$.

\begin{lemma}
  \label{lem:gensuff}
  Let $P$ and $Q$ be distinct elements of $\pauli \setminus \s{\Id}$
  such that $P$ and $Q$ are Hermitian and $PQ=QP$. Then there exists
  $P',Q' \in \pauli$ and $C\in\clifford$ such that $R(P',Q')\in\gens$
  and $R(P,Q) = R(P',Q')C$.
\end{lemma}

\begin{proof}
  Let $P=i^p(P_1\otimes P_2)$ and $Q=i^q(Q_1\otimes Q_2)$ with $P_1,
  P_2, Q_1, Q_2 \in\s{I,X,Y,Z}$. Since $P$ and $Q$ are Hermitian, $p$
  and $q$ must be even. Moreover, by \cref{eq:swappable,eq:minusPauli}
  of \cref{lem:rels}, we can assume without loss of generality that
  $p=q=0$ so that $P=P_1\otimes P_2$ and $Q=Q_1\otimes Q_2$. Now, if
  one of $P_1$, $P_2$, $Q_1$, or $Q_2$ is $I$, then we can use
  \cref{eq:swappable,eq:permutable,eq:minusPauli} of \cref{lem:rels}
  to rewrite $R(P,Q)$ as $R(P',Q')C$ with $C\in\clifford$ and
  $(P',Q')\in\mathcal{T}_1$ as in \cref{def:genset}. If, instead, none
  of $P_1$, $P_2$, $Q_1$, or $Q_2$ are $I$, then we can reason
  similarly to rewrite $R(P,Q)$ as $R(P',Q')C$ with $C\in\clifford$
  and $(P',Q')\in\mathcal{T}_2$.
\end{proof}

\begin{proposition}
  \label{prop:gen}
  Let $V\in\cliffordcs$. Then $V = R_1\cdots R_n C$ where
  $C\in\clifford$ and $R_j\in\gens$ for $j\in [n]$.
\end{proposition}

\begin{proof}
  Let $V\in\cliffordcs$. Then $V$ can be written as $V = C_1 \cdot CS
  \cdot C_2 \cdot CS \cdot \ldots \cdot C_n \cdot CS \cdot C_{n+1}$
  where $C_j\in\clifford$ for $j \in [n+1]$. Since $CS = R(Z\otimes I,
  I\otimes Z)$ we have
  \begin{equation}
  \label{eq:v}
  V = C_1 \cdot R(Z\otimes I, I\otimes Z) \cdot C_2 \cdot R(Z\otimes
  I, I\otimes Z) \cdot \ldots \cdot C_n \cdot R(Z\otimes I, I\otimes
  Z) \cdot C_{n+1}.
  \end{equation}
  Now, by \cref{eq:CliffordCommute} of \cref{lem:rels}, $C_1R(Z\otimes
  I, I\otimes Z) = C_1R(Z\otimes I, I\otimes Z) C_1^\dagger C_1 =
  R(P,Q)C_1$ for some $P,Q\in\pauli$. We can then apply
  \cref{lem:gensuff} to get
  \[
  C_1R(Z\otimes I, I\otimes Z) = R(P,Q)C_1 = R(P',Q')CC_1 = R(P',Q')C'
  \]
  with $C' = CC_1\in\clifford$ and $R(P',Q')\in\gens$. Hence, setting
  $R_1= R(P',Q')$ and $C_2'=C'C_2$, \cref{eq:v} becomes
  \[
  V = R_1 \cdot C_2' \cdot R(Z\otimes I, I\otimes Z) \cdot \ldots
  \cdot C_n \cdot R(Z\otimes I, I\otimes Z) \cdot C_{n+1}
  \]
  and we can proceed recursively to complete the proof.
\end{proof}

\section{The Isomorphism
  \texorpdfstring{$\mbox{SU}(4)\cong \mbox{Spin}(6)$}{SU(4)-Spin(6)}}
\label{sec:iso}

In this section, we describe the exceptional isomorphism $\su(4)\cong
\Spin(6)$ which will allow us to rewrite two-qubit operators as
elements of $\so(6)$. Consider some element $U$ of
$\mbox{SU}(4)$. Then $U$ acts on $\C^4$ by
left-multiplication. Moreover, this action is norm-preserving. Now let
$\s{e_j}$ be the standard orthonormal basis of $\C^4$. From this
basis, we construct an alternative six-component basis using the
\emph{wedge product}.

\begin{definition}[Wedge product]
  \label{def:wedge}
  Let $a\wedge b$ be defined as the \emph{wedge product} of $a$ and
  $b$. Wedge products have the following properties given vectors
  $a,b,c\in \C^n$ and $\alpha,\beta\in\C$:
  \begin{itemize}
  \item Anticommutativity: $a\wedge b = -b \wedge a$.
  \item Associativity: $(a\wedge b)\wedge c= a\wedge (b\wedge c)$.
  \item Bilinearity: $(\alpha a + \beta b)\wedge c = \alpha (a\wedge
    c) + \beta (b\wedge c)$.
  \end{itemize}
  Note that the anticommutation of wedge products implies that
  $a\wedge a=0$. We say that $v_1\wedge\cdots\wedge v_k\in\bigwedge^k
  \C^n$ for $v_j\in\C^n$. To compute the inner product of two wedge
  products $v_1\wedge\cdots\wedge v_k$ and $w_1\wedge\cdots\wedge
  w_k$, we compute
  \[
  \langle v_1\wedge\cdots\wedge v_k, w_1\wedge\cdots\wedge w_k \rangle
  = \det\left(\langle v_q,w_r\rangle\right)
  \]
  where $\langle v_q,w_r\rangle$ is the entry in the $q$-th row and
  $r$-th column of a $k\times k$ matrix.
\end{definition}

\begin{remark}
  The magnitude of a wedge product of $n$ vectors can be thought of as
  the $n$ dimensional volume of the parallelotope constructed from
  those vectors. The orientation of the wedge product defines the
  direction of circulation around that parallelotope by those vectors.
\end{remark}

The wedge product of two vectors in $\C^4$ can be decomposed into a
six-component basis as anticommutativity reduces the 16 potential
wedge products of elements of $\s{e_j}$ to six. We choose this basis
as
\begin{align}
  \label{eq:basis}
  B =
  \s{s_{-,12,34},s_{+,12,34},s_{-,23,14},s_{+,24,13},s_{-,24,13},s_{+,23,14}}
\end{align}
where
\begin{align}
  s_{\pm,ij,kl} = \frac{i^\frac{1\mp 1}{2}}{\sqrt{2}}\left(e_i\wedge
  e_j \pm e_k\wedge e_l\right).
\end{align}
We note that $B$ is an orthonormal basis and we assume that $B$ is
ordered as in \cref{eq:basis}.

\begin{definition}
  \label{def:action}
  Let $U\in \su(4)$ and $\oline{U}$ be its representation in the
  transformed basis. Let $v,w\in\C^4$ with $v\wedge
  w\in\bigwedge^2\C^4$. Then the actions of $U$ and $\oline{U}$ are
  related by
  \[
  \oline{U} (v\wedge w) = (U v)\wedge(U w).
  \]
\end{definition}

To avoid confusion, we use an overline, as in $\oline{O}$, to denote
the $\so(6)$ representation of an operator or set of operators $O$. We
are now equipped to define the transformation from $\su(4)$ to
$\so(6)$.

\begin{definition}
  \label{def:isom}
  Let $U\in \mbox{SU}(4)$ and let $j,k\in[6]$. Then the entry in the
  $j$-th row and $k$-th column of the $\so(6)$ representation
  $\oline{U}$ of $U$ is
  \begin{align}
  \label{eq:rep}
  \oline{U}_{j,k} = \langle B_j, \oline{U} B_k\rangle
  \end{align}
  where $ B_j$ is the $j$-th element in the ordered basis $B$, the
  action of $\oline{U}$ on $B_k$ is defined by
  \cref{def:wedge,def:action}, and the inner product is defined by
  \cref{def:wedge}.
\end{definition}

As an illustration of the process specified in \cref{def:isom} we
explicitly calculate the $\so(6)$ representation of a Clifford+$CS$
operator in \cref{app:calculation}. Moreover, we provide code to
compute this isomorphism for any input with our Mathematica package
\cite{thecode}.

\begin{remark}
  The fact that this isomorphism yields special orthogonal operators
  is ultimately due to the fact that the Dynkin diagrams for the Lie
  algebras of $\su(4)$, $\Spin(6)$, and $\so(6)$ are
  equivalent. However, this fact can be easily illustrated through the
  Euler decomposition of $\su(4)$ \cite{tilma2002generalized}. Direct
  calculation of $\oline{U}$ for the operator
  \[
  U = \begin{bmatrix}
  1 & \cdot & \cdot & \cdot \\
  \cdot & 1 & \cdot & \cdot\\
  \cdot & \cdot & \alpha & \cdot\\
  \cdot & \cdot & \cdot & \alpha^*
  \end{bmatrix}
  \]
  for $|\alpha|=1$ and $\alpha = r+ic$ with $r,c\in\R$ yields
  \[
  \oline{U}= \begin{bmatrix}
  1 & \cdot & \cdot & \cdot & \cdot & \cdot\\
  \cdot & 1 & \cdot & \cdot & \cdot & \cdot\\
  \cdot & \cdot & r & \cdot & \cdot & c\\
  \cdot & \cdot & \cdot & r & c & \cdot\\
  \cdot & \cdot & \cdot & -c & r & \cdot\\
  \cdot & \cdot & -c & \cdot & \cdot & r
  \end{bmatrix}
  \]
  which is explicitly in $\so(6)$. Computation of the other 14 Euler
  angle rotations required for an $\su(4)$ parameterization yields
  similar matrices, likewise in $\so(6)$. Since $\so(6)$ is a group
  under multiplication, the isomorphism applied to any $U\in\su(4)$
  yields $\oline{U}\in\so(6)$.
\end{remark}

We close this section by explicitly calculating the $\so(6)$
representation of each of the generators of $\cliffordcs$. We multiply
the generators by overall phase factors to ensure that each operator
has determinant one, and furthermore that single-qubit operators have
determinant one on their single-qubit subspace. Later, when referring
to gates or their $\so(6)$ representation, we omit overall phases for
readability.

\begin{proposition}
  \label{prop:imcliffords}
  The image of the generators of $\clifford$ in $\so(6)$ are
  \[
  \begin{array}{rclcrcl}
  \oline{(\omega^\dagger S)\otimes\Id} & = &  \begin{bmatrix}
  \cdot & -1 & \cdot & \cdot & \cdot & \cdot\\
  1 & \cdot & \cdot & \cdot & \cdot & \cdot\\
  \cdot & \cdot & 1 & \cdot & \cdot & \cdot\\
  \cdot & \cdot & \cdot & 1 & \cdot & \cdot\\
  \cdot & \cdot & \cdot & \cdot & 1 & \cdot\\
  \cdot & \cdot & \cdot & \cdot & \cdot & 1
  \end{bmatrix}, & \qquad &
  \oline{ \Id\otimes (\omega^\dagger S)} & = & \begin{bmatrix}
  1 & \cdot & \cdot & \cdot & \cdot & \cdot\\
  \cdot & 1 & \cdot & \cdot & \cdot & \cdot\\
  \cdot & \cdot & 1 & \cdot & \cdot & \cdot\\
  \cdot & \cdot & \cdot & \cdot & -1 & \cdot\\
  \cdot & \cdot & \cdot & 1 & \cdot & \cdot\\
  \cdot & \cdot & \cdot & \cdot & \cdot & 1
  \end{bmatrix}, \\
  ~ & ~ \\
  \oline{(i H)\otimes\Id} & = & \begin{bmatrix}
  \cdot & \cdot & 1 & \cdot & \cdot & \cdot\\
  \cdot & -1 & \cdot & \cdot & \cdot & \cdot\\
  1 & \cdot & \cdot & \cdot & \cdot & \cdot\\
  \cdot & \cdot & \cdot & 1 & \cdot & \cdot\\
  \cdot & \cdot & \cdot & \cdot & 1 & \cdot\\
  \cdot & \cdot & \cdot & \cdot & \cdot & 1
  \end{bmatrix}, & \qquad &
  \oline{\Id\otimes(i H)} & = & \begin{bmatrix}
  1 & \cdot & \cdot & \cdot & \cdot & \cdot\\
  \cdot & 1 & \cdot & \cdot & \cdot & \cdot\\
   & \cdot & 1 & \cdot & \cdot & \cdot\\
  \cdot & \cdot & \cdot & \cdot & \cdot & 1\\
  \cdot & \cdot & \cdot & \cdot & -1 & \cdot\\
  \cdot & \cdot & \cdot & 1 & \cdot & \cdot
  \end{bmatrix},
  \end{array}
  \]
  \[
  \begin{array}{rcl}
  \oline{\omega^\dagger CZ} & = & \begin{bmatrix}
  \cdot & -1 & \cdot & \cdot & \cdot & \cdot\\
  1 & \cdot & \cdot & \cdot & \cdot & \cdot\\
  \cdot & \cdot & \cdot & \cdot & \cdot & -1\\
  \cdot & \cdot & \cdot & \cdot & -1 & \cdot\\
  \cdot & \cdot & \cdot & 1 & \cdot & \cdot\\
  \cdot & \cdot & 1 & \cdot & \cdot & \cdot
  \end{bmatrix}.
  \end{array}
  \]
\end{proposition}

\begin{proposition}
  \label{prop:imgens}
  The elements of $\oline\gens$ are given in \cref{fig:olineelems}.
\end{proposition}

\begin{figure}
  \centering
  \[
  \frac{1}{\sqrt 2}\begin{bmatrix}
    1 & \cdot & \cdot & -1 & \cdot & \cdot\\
    \cdot & 1 & -1 & \cdot & \cdot & \cdot\\
    \cdot & 1 & 1 & \cdot & \cdot & \cdot\\
    1 & \cdot & \cdot & 1 & \cdot & \cdot\\
    \cdot & \cdot & \cdot & \cdot & 1 & -1\\
    \cdot & \cdot & \cdot & \cdot & 1 & 1
  \end{bmatrix}\qquad
  \frac{1}{\sqrt 2}\begin{bmatrix}
    1 & \cdot & 1 & \cdot & \cdot & \cdot\\
    \cdot & 1 & \cdot & \cdot & -1 & \cdot\\
    -1 & \cdot & 1 & \cdot & \cdot & \cdot\\
    \cdot & \cdot & \cdot & 1 & \cdot & 1\\
    \cdot & 1 & \cdot & \cdot & 1 & \cdot\\
    \cdot & \cdot & \cdot & -1 & \cdot & 1
  \end{bmatrix}\qquad
  \frac{1}{\sqrt 2}\begin{bmatrix}
    1 & -1 & \cdot & \cdot & \cdot & \cdot\\
    1 & 1 & \cdot & \cdot & \cdot & \cdot\\
    \cdot & \cdot & 1 & \cdot & \cdot & -1\\
    \cdot & \cdot & \cdot & 1 & -1 & \cdot\\
    \cdot & \cdot & \cdot & 1 & 1 & \cdot\\
    \cdot & \cdot & 1 & \cdot & \cdot & 1
  \end{bmatrix}
  \]
  
  \[
  \frac{1}{\sqrt 2}\begin{bmatrix}
    1 & \cdot & 1 & \cdot & \cdot & \cdot\\
    \cdot & 1 & \cdot & \cdot & \cdot & -1\\
    -1 & \cdot & 1 & \cdot & \cdot & \cdot\\
    \cdot & \cdot & \cdot & 1 & -1 & \cdot\\
    \cdot & \cdot & \cdot & 1 & 1 & \cdot\\
    \cdot & 1 & \cdot & \cdot & \cdot & 1
  \end{bmatrix}\qquad
  \frac{1}{\sqrt 2}\begin{bmatrix}
    1 & -1 & \cdot & \cdot & \cdot & \cdot\\
    1 & 1 & \cdot & \cdot & \cdot & \cdot\\
    \cdot & \cdot & 1 & \cdot & -1 & \cdot\\
    \cdot & \cdot & \cdot & 1 & \cdot & 1\\
    \cdot & \cdot & 1 & \cdot & 1 & \cdot\\
    \cdot & \cdot & \cdot & -1 & \cdot & 1
  \end{bmatrix}\qquad
  \frac{1}{\sqrt 2}\begin{bmatrix}
    1 & -1 & \cdot & \cdot & \cdot & \cdot\\
    1 & 1 & \cdot & \cdot & \cdot & \cdot\\
    \cdot & \cdot & 1 & -1 & \cdot & \cdot\\
    \cdot & \cdot & 1 & 1 & \cdot & \cdot\\
    \cdot & \cdot & \cdot & \cdot & 1 & -1\\
    \cdot & \cdot & \cdot & \cdot & 1 & 1
  \end{bmatrix}
  \]
  
  \[
  \frac{1}{\sqrt 2}\begin{bmatrix}
    1 & \cdot & \cdot & \cdot & \cdot & -1\\
    \cdot & 1 & -1 & \cdot & \cdot & \cdot\\
    \cdot & 1 & 1 & \cdot & \cdot & \cdot\\
    \cdot & \cdot & \cdot & 1 & -1 & \cdot\\
    \cdot & \cdot & \cdot & 1 & 1 & \cdot\\
    1 & \cdot & \cdot & \cdot & \cdot & 1
  \end{bmatrix}\qquad
  \frac{1}{\sqrt 2}\begin{bmatrix}
    1 & \cdot & \cdot & \cdot & -1 & \cdot\\
    \cdot & 1 & -1 & \cdot & \cdot & \cdot\\
    \cdot & 1 & 1 & \cdot & \cdot & \cdot\\
    \cdot & \cdot & \cdot & 1 & \cdot & 1\\
    1 & \cdot & \cdot & \cdot & 1 & \cdot\\
    \cdot & \cdot & \cdot & -1 & \cdot & 1
  \end{bmatrix}\qquad
  \frac{1}{\sqrt 2}\begin{bmatrix}
    1 & \cdot & 1 & \cdot & \cdot & \cdot\\
    \cdot & 1 & \cdot & -1 & \cdot & \cdot\\
    -1 & \cdot & 1 & \cdot & \cdot & \cdot\\
    \cdot & 1 & \cdot & 1 & \cdot & \cdot\\
    \cdot & \cdot & \cdot & \cdot & 1 & -1\\
    \cdot & \cdot & \cdot & \cdot & 1 & 1
  \end{bmatrix}
  \]
  
  \[
  \frac{1}{\sqrt 2}\begin{bmatrix}
    1 & \cdot & \cdot & 1 & \cdot & \cdot\\
    \cdot & 1 & \cdot & \cdot & 1 & \cdot\\
    \cdot & \cdot & 1 & \cdot & \cdot & 1\\
    -1 & \cdot & \cdot & 1 & \cdot & \cdot\\
    \cdot & -1 & \cdot & \cdot & 1 & \cdot\\
    \cdot & \cdot & -1 & \cdot & \cdot & 1
  \end{bmatrix}\qquad
  \frac{1}{\sqrt 2}\begin{bmatrix}
    1 & \cdot & \cdot & -1 & \cdot & \cdot\\
    \cdot & 1 & \cdot & \cdot & \cdot & 1\\
    \cdot & \cdot & 1 & \cdot & 1 & \cdot\\
    1 & \cdot & \cdot & 1 & \cdot & \cdot\\
    \cdot & \cdot & -1 & \cdot & 1 & \cdot\\
    \cdot & -1 & \cdot & \cdot & \cdot & 1
  \end{bmatrix}\qquad
  \frac{1}{\sqrt 2}\begin{bmatrix}
    1 & \cdot & \cdot & \cdot & \cdot & 1\\
    \cdot & 1 & \cdot & \cdot & -1 & \cdot\\
    \cdot & \cdot & 1 & 1 & \cdot & \cdot\\
    \cdot & \cdot & -1 & 1 & \cdot & \cdot\\
    \cdot & 1 & \cdot & \cdot & 1 & \cdot\\
    -1 & \cdot & \cdot & \cdot & \cdot & 1
  \end{bmatrix}
  \]
  
  \[
  \frac{1}{\sqrt 2}\begin{bmatrix}
    1 & \cdot & \cdot & \cdot & -1 & \cdot\\
    \cdot & 1 & \cdot & 1 & \cdot & \cdot\\
    \cdot & \cdot & 1 & \cdot & \cdot & 1\\
    \cdot & -1 & \cdot & 1 & \cdot & \cdot\\
    1 & \cdot & \cdot & \cdot & 1 & \cdot\\
    \cdot & \cdot & -1 & \cdot & \cdot & 1
  \end{bmatrix}\qquad
  \frac{1}{\sqrt 2}\begin{bmatrix}
    1 & \cdot & \cdot & \cdot & 1 & \cdot\\
    \cdot & 1 & \cdot & \cdot & \cdot & 1\\
    \cdot & \cdot & 1 & 1 & \cdot & \cdot\\
    \cdot & \cdot & -1 & 1 & \cdot & \cdot\\
    -1 & \cdot & \cdot & \cdot & 1 & \cdot\\
    \cdot & -1 & \cdot & \cdot & \cdot & 1
  \end{bmatrix}\qquad
  \frac{1}{\sqrt 2}\begin{bmatrix}
    1 & \cdot & \cdot & \cdot & \cdot & 1\\
    \cdot & 1 & \cdot & 1 & \cdot & \cdot\\
    \cdot & \cdot & 1 & \cdot & 1 & \cdot\\
    \cdot & -1 & \cdot & 1 & \cdot & \cdot\\
    \cdot & \cdot & -1 & \cdot & 1 & \cdot\\
    -1 & \cdot & \cdot & \cdot & \cdot & 1
  \end{bmatrix}
  \]
  \caption{The 15 elements of $\oline\gens$.\label{fig:olineelems}}
\end{figure}
  
\section{Exact Synthesis}
\label{sec:synth}

In this section, we leverage the isomorphism $\su(4)\cong \Spin(6)$
described in the previous section to find optimal decompositions for
the elements of $\cliffordcs$. We will be working extensively with the
matrix group
\begin{equation}
  \label{eq:somatrix}
  \matrices=\s{\frac{1}{\sqrt{2}^k}M\in\so(6)~;~ k\in\N,
    M\in\Z^{6\times 6}}.
\end{equation}
Note that $\matrices\subseteq \so(6)$. Our interest in $\matrices$
stems from the following observation.

\begin{proposition}
  \label{prop:CliffCSinD}
  We have $\oline\cliffordcs\subseteq\matrices$.
\end{proposition}

\begin{proof}
  The property holds for the generators of $\oline\cliffordcs$ by
  \cref{prop:imcliffords,prop:imgens}.
\end{proof}

In the remainder of this section, we prove the converse of
\cref{prop:CliffCSinD} by defining an algorithm which inputs an
element of $\matrices$ and outputs a product of generators. We start
by introducing a few notions that are useful in discussing the
elements of $\matrices$.

\begin{definition}
  \label{def:lde}
  Let $V\in\matrices$. We say that $\ell\in\N$ is a \emph{denominator
    exponent} of $V$ if $\sqrt{2}^\ell V\in\Z^{6\times 6}$. The least
  such $\ell$ is the \emph{least denominator exponent} of $V$, which
  we denote by $\lde(V)$.
\end{definition}

\begin{lemma}
  \label{lem:CScountlower}
  Let $U\in\cliffordcs$ and suppose that $\lde(\oline{U})=k$. Then any
  Clifford+$CS$ circuit for $U$ has $CS$-count at least $k$.
\end{lemma}

\begin{proof}
  The only generators with a factor of $1/\sqrt{2}$ in their $\so(6)$
  representation are the elements of $\gens$. Thus, for a least
  denominator exponent of $k$ there must be at least $k$ of these
  operators, each of which requires a single $CS$ gate.
\end{proof}

\begin{definition}
  \label{def:kparity}
  Let $V\in\matrices$ and let $\ell$ be a denominator exponent of
  $V$. The \emph{$\ell$-residue} of $V$ is the binary matrix
  $\rho_\ell(V)\in \Z_2^{6\times 6}$ defined by
  \[
  (\rho_\ell(V))_{i,j} = \rho((\sqrt{2}^\ell V)_{i,j})
  \]
  where $\rho :\Z \to \Z_2$ is the canonical (parity) homomorphism.
\end{definition}

The residue matrices introduced in \cref{def:kparity} are important in
the definition of the exact synthesis algorithm. Indeed, the
$\ell$-residue of a Clifford+$CS$ operator $U$ determines the element
of $\gens$ to use in order to reduce the least denominator exponent of
$U$ (although not uniquely, as we discuss below). Similar residue
matrices are used in the study of other fault-tolerant circuits
\cite{restricted,ma-remarks}.

Recall that if $A$ is a set, then a \emph{partition} of $A$ is a
collection of disjoint nonempty subsets of $A$ whose union is equal to
$A$. The set of all partitions of a set $A$ is denoted
$\mathscr{B}_A$. Let $p$ and $p'$ be two partitions of $A$. If every
element of $p$ is a subset of an element of $p'$ then we say that $p'$
is \emph{coarser} than $p$ and that $p$ is \emph{finer} than $p'$.

\begin{definition}
  \label{def:pattern}
  Let $N\in \Z_2^{6\times 6}$ be a binary matrix with rows
  $r_1,\ldots, r_6$ and let $p=\s{p_1,\ldots,p_q}$ be a partition of
  the set $[6]$. Then $N$ has the \emph{pattern} $p$ if for any
  $p_j$ in $p$ and any $j_1,j_2\in p_j$ we have $r_{j_1}=r_{j_2}$. In
  this case we also say that $N$ has a \emph{$|p_1|\times \ldots
    \times |p_q|$ pattern}.
\end{definition}

\begin{definition}
  \label{def:patternmap}
  Let $V\in\matrices$ with $\lde(V)=\ell$. We define the pattern map
  $\partition: \matrices\rightarrow \mathscr{B}_{[6]}$ as the function
  which maps $V$ to the pattern of $\rho_\ell (V)$. We say that
  $p=\partition(V)$ is the pattern of $V$. If $V_1$ and $V_2$ are two
  elements of $\matrices$, we say that $V_1$ is \emph{finer} than
  $V_2$ or that $V_2$ is \emph{coarser} than $V_1$ if these statements
  hold for $\partition(V_1)$ and $\partition(V_2)$.
\end{definition}

\begin{remark}
  In a slight abuse of notation, we extend the pattern map to any
  valid representation of a Clifford+$CS$ operator. Given a
  Clifford+$CS$ operator with $\su(4)$ representation $U$ which can be
  written as a word $W$ over the generators and with $\so(6)$
  representation $\oline{U}$, we set $\partition(U) = \partition(W) =
  \partition(\oline{U})$. This extension is unambiguous after fixing
  our transformation from $\su(4)$ to $\so(6)$, as $\partition$ is
  insensitive to relative phase changes in $U$. We incorporate all
  relational notions described in \cref{def:patternmap} in this
  extension.
\end{remark}

We now analyze the image in $\so(6)$ of certain subsets of
$\cliffordcs$. We start by showing that the image of the Clifford
group $\clifford$ is exactly the collection of elements of $\matrices$
with least denominator 0. In other words, $\oline\clifford$ is the
group of $6$-dimensional signed permutation matrices.

\begin{lemma}
  \label{lem:CliffinD}
  Let $V\in\matrices$. Then $\lde(V)=0$ if and only if
  $V\in\oline\clifford$.
\end{lemma}

\begin{proof}
  The least denominator exponent of $\oline{H\otimes I}$,
  $\oline{I\otimes H}$, $\oline{S\otimes I}$, $\oline{I\otimes S}$,
  and $\oline{CZ}$ is 0. Thus, if $U\in\clifford$ then
  $\lde(\oline{U})=0$. For the converse, let $C_1$ and $C_2$ be the
  Clifford operators $(\omega^\dagger S)\otimes I$ and $(H \otimes H)
  (\omega^\dagger CZ)(Z \otimes Z)$, respectively. Then
  \[
  \oline{C_1} =
  \begin{bmatrix}
  \cdot & -1 & \cdot & \cdot & \cdot & \cdot\\
  1 & \cdot & \cdot & \cdot & \cdot & \cdot\\
  \cdot & \cdot & 1 & \cdot & \cdot & \cdot\\
  \cdot & \cdot & \cdot & 1 & \cdot & \cdot\\
  \cdot & \cdot & \cdot & \cdot & 1 & \cdot\\
  \cdot & \cdot & \cdot & \cdot & \cdot & 1
  \end{bmatrix}\quad\mbox{and}\quad
  \oline{C_2} =
  \begin{bmatrix}
  \cdot & \cdot & \cdot & \cdot & \cdot & -1\\
  1 & \cdot & \cdot & \cdot & \cdot & \cdot\\
  \cdot & 1 & \cdot & \cdot & \cdot & \cdot\\
  \cdot & \cdot & 1 & \cdot & \cdot & \cdot\\
  \cdot & \cdot & \cdot & 1 & \cdot & \cdot\\
  \cdot & \cdot & \cdot & \cdot & 1 & \cdot
  \end{bmatrix}.
  \]
  The operators $\oline{C_1}$ and $\oline{C_2}$ generate
  $\s{V\in\matrices ~;~ \lde(V)=0}$. Hence, if $V\in\matrices$ and
  $\lde(V)=0$ then $V$ can be expressed as a product of the image of
  Clifford gates.
\end{proof}

\begin{lemma}
  \label{lem:GensinD}
  Let $V\in\matrices$. Then $\lde(V)=1$ if and only if $V=\oline{RC}$
  for some $R\in\gens$ and some $C\in\clifford$. Furthermore, $V$ has
  a $2\times2\times2$ pattern.
 \end{lemma}

\begin{proof}
  The rows of $V$ have unit norm and are pairwise orthogonal. Hence,
  up to a signed permutation of rows and columns, there is only one
  such matrix, e.g.,
  \begin{align}
    \label{eq:k1denom}
    \frac{1}{\sqrt{2}}\begin{bmatrix}
    1 & -1 & \cdot & \cdot & \cdot & \cdot\\
    1 & 1 & \cdot & \cdot & \cdot & \cdot\\
    \cdot & \cdot & 1 & -1 & \cdot & \cdot\\
    \cdot & \cdot & 1 & 1 & \cdot & \cdot\\
    \cdot & \cdot & \cdot & \cdot & 1 & -1\\
    \cdot & \cdot & \cdot & \cdot & 1 & 1
    \end{bmatrix} = \oline{\gens}_6.
  \end{align}
  By \cref{prop:gen} the proof is complete, since Clifford operators
  correspond to signed permutations by \cref{lem:CliffinD}.
\end{proof}

\begin{lemma}
  \label{lem:zerorows}
  Let $V\in\matrices$ with $\lde(V)=k\geq 2$. Then $V$ has either a
  $2\times2\times2$ or $2\times4$ pattern.
\end{lemma}

\begin{proof}
  Let $V\in\matrices$. Since $V$ is orthogonal, we have $V^\dagger V=
  I$. Hence, $(\sqrt{2}{}^kV)^\dagger(\sqrt{2}{}^k V)=2^kI$. Since
  $k\geq 2$, this implies that the inner product of any column of
  $\sqrt{2}{}^k V$ with itself is congruent to 0 modulo 4. Similarly,
  the inner product of two distinct columns $\sqrt{2}{}^k V$ is
  congruent to 0 modulo 4. Letting, $M=\rho_k(V)$, we then have the
  column relations
  \begin{align}
    \sum_{l} M_{lm}^2 &= 0\mod 4\label{eq:columnodd}\\
    \sum_{l} M_{lm} M _{ln} &= 0\mod 2 \mbox{ for } m\neq n\label{eq:rowpair}
  \end{align}
  as well as analogous row relations. For $x\in\Z$, $x^2=0\mod 4$ if
  and only if $x=0\mod2$. Hence, there must be exactly zero or four
  odd entries in every column (or row) of $M$ by
  \cref{eq:columnodd}. By \cref{eq:rowpair}, we see that the inner
  product of any two distinct rows must be even. Up to a permutation
  of rows and columns, we can then deduce that $M$ is one of the two
  matrices below, which completes the proof.
  \begin{align}
    \label{eq:k2denom}
    \begin{bmatrix}
      1 & 1 & 1 & 1 & \cdot & \cdot\\
      1 & 1 & 1 & 1 & \cdot & \cdot\\
      1 & 1 & 1 & 1 & \cdot & \cdot\\
      1 & 1 & 1 & 1 & \cdot & \cdot\\
      \cdot & \cdot & \cdot & \cdot & \cdot & \cdot\\
      \cdot & \cdot & \cdot & \cdot & \cdot & \cdot
    \end{bmatrix}\quad\mbox{or}\quad
    \begin{bmatrix}
      1 & 1 & 1 & 1 & \cdot & \cdot\\
      1 & 1 & 1 & 1 & \cdot & \cdot\\
      1 & 1 & \cdot & \cdot & 1 & 1\\
      1 & 1 & \cdot & \cdot & 1 & 1\\
      \cdot & \cdot & 1 & 1 & 1 & 1\\
      \cdot & \cdot & 1 & 1 & 1 & 1
    \end{bmatrix}
  \end{align}
\end{proof}

\begin{corollary}
  \label{cor:rowpair}
  Let $V\in\matrices$ with $\lde(V)=k\geq 1$. Then $V$ has either a
  $2\times2\times2$ or $2\times4$ pattern.
\end{corollary}

\begin{lemma}
  \label{lem:finer}
  Let $V\in\matrices$ and assume that $\lde(V)=k\geq 1$. If
  $\oline{R}\in\oline\gens$ is finer than $V$, then
  $\lde(\oline{R}^\Trans V) = k-1$.
\end{lemma}

\begin{proof}
  For simplicity, we assume that $\partition(\oline{R}) = \s{\s{1,2},
    \s{3,4}, \s{5,6}}$. The cases in which $\partition(\oline{R})$ is
  another pattern are treated similarly. For $j\in[6]$, let $r_j$
  denote the rows of $\sqrt{2}^k V$. Since $\partition(V)$ is coarser
  than $\partition(\oline{R})$, we have $r_1\equiv r_2$, $r_3\equiv
  r_4$, $r_5\equiv r_6$ modulo 2. This implies that $r_1 \pm r_2
  \equiv r_3 \pm r_4 \equiv r_5 \pm r_6 \equiv 0$ modulo 2. Hence
  \[
  \oline{R}^\Trans V = \frac{1}{\sqrt{2}^{k+1}}
  \begin{bmatrix}
    1 & 1 & \cdot & \cdot & \cdot & \cdot\\
    -1 & 1 & \cdot & \cdot & \cdot & \cdot\\
    \cdot & \cdot & 1 & 1 & \cdot & \cdot\\
    \cdot & \cdot & -1 & 1 & \cdot & \cdot\\
    \cdot & \cdot & \cdot & \cdot & 1 & 1\\
    \cdot & \cdot & \cdot & \cdot & -1 & 1
  \end{bmatrix}
  \begin{bmatrix}
    r_1 \\
    r_2 \\
    r_3 \\
    r_4 \\
    r_5 \\
    r_6 
  \end{bmatrix}
  =
  \frac{1}{\sqrt{2}^{k+1}}  
  \begin{bmatrix}
    r_1 - r_2 \\
    r_1 + r_2 \\
    r_3 - r_4 \\
    r_3 + r_4 \\
    r_5 - r_6 \\
    r_5 + r_6 
  \end{bmatrix}
  =
  \frac{1}{\sqrt{2}^{k-1}}  
  \begin{bmatrix}
    r_1' \\
    r_2' \\
    r_3' \\
    r_4' \\
    r_5' \\
    r_6'
  \end{bmatrix}.      
  \]
  where each $r_j'$ is a vector of integers.
\end{proof}

\begin{lemma}
  \label{lem:denomreduce}
  Let $V\in\matrices$ with $\lde(V)\geq 1$. Then there exists
  $R\in\gens$ such that $\lde(\oline{R}^\Trans V) = \lde(V)-1$.
\end{lemma}

\begin{proof}
  By inspection of \cref{fig:olineelems} we see that for every
  $2\times 2\times 2$ pattern $q$ there exists $R\in\gens$ such that
  $\partition(\oline{R})=q$. As a result, if $\partition(V)$ is a
  $2\times 2 \times 2$ or a $2\times 4$ pattern, then there exists
  $R\in\gens$ such that $\oline{R}$ has a pattern finer than
  $\partition(V)$. By \cref{cor:rowpair}, $\partition(V)$ is in fact a
  $2\times 2\times 2$ row-pattern or a $2\times 4$ row-pattern and
  thus there exists $R\in\gens$ such that $\oline{R}$ is finer than
  $V$. We can then conclude by \cref{lem:finer}.
\end{proof}

\begin{theorem}
  \label{thm:DisCliffCS}
  We have $\oline{\cliffordcs}=\matrices$.
\end{theorem}

\begin{proof}
  $\oline{\cliffordcs}\subseteq\matrices$ by \cref{prop:CliffCSinD}.
  We now show $\matrices\subseteq\oline{\cliffordcs}$. Let
  $V\in\matrices$. We proceed by induction on the least denominator
  exponent of $V$. If $\lde(V)=0$ then, by \cref{lem:CliffinD},
  $V\in\oline\clifford$ and therefore $V\in \oline\cliffordcs$. Now if
  $\lde(V)>0$, let $R$ be the element of $\gens$ with the lowest index
  such that $\lde(\oline{R}^\Trans V)=k-1$. Such an element exists by
  \cref{lem:denomreduce}. By the induction hypothesis we have
  $\oline{R}^\Trans V\in\oline\cliffordcs$ which implies that
  $\oline{R}(\oline{R}^\Trans V) = V\in\oline\cliffordcs$.
\end{proof}

The proof of \cref{thm:DisCliffCS} provides an algorithm to decompose
an arbitrary element of $\oline\cliffordcs$ into a product of elements
of $\oline \gens$, followed by an element of $\oline\clifford$. In the
proof, there is freedom in choosing the element of $\oline\gens$ used
to reduce $\lde(\oline{V})$. If there is more than one generator with
a finer pattern than $\oline{V}$, we must make a choice. The ordering
imposed on $\gens$ in \cref{sec:gens} is used to make this choice in a
uniform manner: we always choose the element of $\gens$ of lowest
index. As a result, the exact synthesis algorithm becomes
deterministic. The ambiguity in the choice of generator is a
consequence of the relations given in \cref{lem:rels}. In particular,
we have
\[
R(P,L)R(P,Q)=R(P,Q)R(P,iQL)=R(P,iQL)R(P,L)
\]
and these three distinct sequences of generators denote the same
operator. This is the source of the three-fold ambiguity in choosing a
finer $2\times2\times2$ pattern for a given $2\times 4$ pattern.

We will sometimes refer to the association between elements of $\gens$
and patterns used in the exact synthesis algorithm of
\cref{thm:DisCliffCS} as the \emph{first finer partition} association,
or FFP for short. The association is explicitly described
\cref{table:FFP}.

\begin{table}[t]
  \begin{center}
    \begin{tabular}{c|l}
      \textbf{Generator} & \textbf{Associated Patterns Under First Finer Partition (FFP)}\\
      \hline
      $R(X \otimes I,I \otimes X)$ & $\s{\s{1,4},\s{2,3},\s{5,6}},\s{\s{1,4},\s{2,3,5,6}},\s{\s{2,3},\s{1,4,5,6}},\s{\s{5,6},\s{1,2,3,4}}$\\
      $R(Y \otimes I,I \otimes Y)$ & $\s{\s{1,3},\s{2,5},\s{4,6}},\s{\s{1,3},\s{2,4,5,6}},\s{\s{2,5},\s{1,3,4,6}},\s{\s{4,6},\s{1,2,3,5}}$\\
      $R(Z \otimes I,I \otimes Z)$ & $\s{\s{1,2},\s{3,6},\s{4,5}},\s{\s{1,2},\s{3,4,5,6}},\s{\s{3,6},\s{1,2,4,5}},\s{\s{4,5},\s{1,2,3,6}}$\\
      $R(Y \otimes I,I \otimes Z)$ & $\s{\s{1,3},\s{2,6},\s{4,5}},\s{\s{2,6},\s{1,3,4,5}}$\\
      $R(Z \otimes I,I \otimes Y)$ & $\s{\s{1,2},\s{3,5},\s{4,6}},\s{\s{3,5},\s{1,2,4,6}}$\\
      $R(Z \otimes I,I \otimes X)$ & $\s{\s{1,2},\s{3,4},\s{5,6}},\s{\s{3,4},\s{1,2,5,6}}$\\
      $R(X \otimes I,I \otimes Z)$ & $\s{\s{1,6},\s{2,3},\s{4,5}},\s{\s{1,6},\s{2,3,4,5}}$\\
      $R(X \otimes I,I \otimes Y)$ & $\s{\s{1,5},\s{2,3},\s{4,6}},\s{\s{1,5},\s{2,3,4,6}}$\\
      $R(Y \otimes I,I \otimes X)$ & $\s{\s{1,3},\s{2,4},\s{5,6}},\s{\s{2,4},\s{1,3,5,6}}$\\
      $R(X \otimes X,Y \otimes Y)$ & $\s{\s{1,4},\s{2,5},\s{3,6}}$\\
      $R(X \otimes X,Z \otimes Y)$ & $\s{\s{1,4},\s{2,6},\s{3,5}}$\\
      $R(Z \otimes X,Y \otimes Y)$ & $\s{\s{1,6},\s{2,5},\s{3,4}}$\\
      $R(Y \otimes X,X \otimes Y)$ & $\s{\s{1,5},\s{2,4},\s{3,6}}$\\
      $R(Z \otimes X,X \otimes Y)$ & $\s{\s{1,5},\s{2,6},\s{3,4}}$\\
      $R(Y \otimes X,Z \otimes Y)$ & $\s{\s{1,6},\s{2,4},\s{3,5}}$\\
    \end{tabular}
    \caption{The elements of $\gens$ and the explicit row patterns
      they are associated with under FFP. \label{table:FFP}}
  \end{center}
\end{table}

\begin{theorem}
  \label{thm:nf}
  If $U$ is a Clifford+$CS$ operator such that $\lde(\oline{U})=k$,
  then $U$ can be represented by a Clifford+$CS$ circuit of $CS$-count
  $k$. This circuit is optimal in $CS$-count and can be constructed in
  $\mathcal{O}(k)$ arithmetic operations.
\end{theorem}

\begin{proof}
  Let $U$ be as stated. If $k=0$, then $\oline{U}$ belongs to
  $\oline{C}$ and $U$ is therefore a Clifford. If $k>0$, then as in
  \cref{thm:DisCliffCS}, there is a unique $R_k\in\gens$ given by FFP
  such that $\lde(\oline{R}_k^\Trans\oline{U})=k-1$. By induction on
  the least denominator exponent, we have a deterministic synthesis
  algorithm to find a sequence such that
  \[
  \oline{U}=\oline{R}_k\cdots\oline{R}_1\cdot\oline{C}
  \]
  which then implies that $U = R_k \cdots R_1 C$. Each of these $k$ steps involves a constant number of basic arithmetic operations. This circuit has
  $CS$-count $k$, which is optimal by \cref{lem:CScountlower}.
\end{proof}

Our Mathematica package \cite{thecode} implements the algorithm
referred to in \cref{thm:nf} as well as a significant amount of other
tools for two-qubit Clifford + $CS$ circuits. Testing of the
performance of this algorithm on a modest device is presented in
\cref{table:performance}.

\begin{table}[b]
\begin{center}
	\begin{tabular}{c|c|c}
		$CS$-count & Mean Time (s) & Std. Dev. (s)\\
		\hline
		10 & 0.0138 & 0.0044 \\
		100 & 0.0281 & 0.0051  \\ 
		1000 & 0.1135 & 0.0091 \\
		10000 & 1.1883 & 0.0897
	\end{tabular}
	\caption{Performance of the algorithm (in seconds) of
          \cref{thm:nf} as implemented in our Mathematica code
          \cite{thecode}. Each run has constant overhead from
          computing the $\so(6)$ representation for each
          unitary. Deviations from linearity are due to arithmetic
          operations on increasingly large integers. Each mean and
          standard deviation is computed using a sample of 1000 runs
          with pseudorandomly generated operators known to have the
          given minimal $CS$-count. Times are measured using
          Mathematica's in-built \texttt{AbsoluteTiming}
          function. Computations performed on a laptop with an
          Intel(R) Core(TM) i7 CPU running at 2.6 GHz with 6 cores and
          16 GB of RAM runnning macOS Catalina version
          10.15.7.\label{table:performance}}
\end{center}
\end{table}

\section{Normal Forms}
\label{sec:nfs}

In the previous section, we introduced a synthesis algorithm for
Clifford+$CS$ operators. The algorithm takes as input a Clifford+$CS$
matrix and outputs a circuit for the corresponding operator. The
circuit produced by the synthesis algorithm is a word over the
alphabet $\gens\cup\clifford$. Because the algorithm is deterministic,
the word it associates to each operator can be viewed as a normal form
for that operator. In the present section, we use the language of
automata to give a detailed description of the structure of these
normal forms. We include the definitions of some basic concepts from
the theory of automata for completeness. The reader looking for
further details is encouraged to consult \cite{sipser}.

\subsection{Automata}
\label{ssec:automata}

In what follows we sometimes refer to a finite set $\Sigma$ as an
\emph{alphabet}. In such a context, the elements of $\Sigma$ are
referred to as \emph{letters}, $\Sigma^*$ denotes the set of
\emph{words} over $\Sigma$ (which includes the empty word
$\varepsilon$), and the subsets of $\Sigma^*$ are called
\emph{languages over $\Sigma$}. If $w\in\Sigma^*$ is a word over the
alphabet $\Sigma$, we write $|w|$ for the \emph{length} of
$w$. Finally, if $L$ and $L'$ are two languages over an alphabet
$\Sigma$ then their \emph{concatenation} $L\circ L'$ is defined as
$L\circ L' = \s{ww'~;~ w\in L \mbox{ and } w'\in L'}$.

\begin{definition}
  \label{def:automaton}
  A \emph{nondeterministic finite automaton} is a 5-tuple $\automaton$
  where $\Sigma$ and $Q$ are finite sets, $\In$ and $\Fin$ are subsets
  of $Q$, and $\delta:Q\times (\Sigma\cup\s{\varepsilon}) :\to
  \mathscr{P}(Q)$ is a function whose codomain is the power set of
  $Q$. We call $\Sigma$ the \emph{alphabet}, $Q$ the set of
  \emph{states}, $\In$ and $\Fin$ the sets of \emph{initial} and
  \emph{final} states, and $\delta$ the \emph{transition function}.
\end{definition}

\begin{remark}
  \cref{def:automaton} is slightly non-standard. Indeed, automata are
  typically defined as having a single initial state, rather than a
  collection of them. One can then think of \cref{def:automaton} as
  introducing a collection of automata: one for each element of
  $\In$. Alternatively, \cref{def:automaton} can also be recovered
  from the usual definition by assuming that every automaton in the
  sense of \cref{def:automaton} in fact has a single initial state
  $s_0$ related to the elements of $\In$ by
  $\delta(s_0,\varepsilon)=\In$. We chose to introduce automata as in
  \cref{def:automaton} because this results in a slightly cleaner
  presentation.
\end{remark}

It is common to define an automaton $A=\automaton$ by specifying a
directed labelled graph called the \emph{state graph} of $A$. The
vertices of the graph are labelled by states and there is an edge
labelled by a letter $w\in \Sigma$ between vertices labelled $q$ and
$q'$ if $q'\in\delta(q,w)$. The initial and final states are
distinguished using arrows and double lines, respectively. For
brevity, parallel edges are drawn only once, with their labels
separated by a comma. 

\begin{example}
  \label{ex:aut}
  The state graph for a nondeterministic finite automaton
  $A=(\Sigma,Q,\delta,\In,\Fin)$ is depicted below.
  
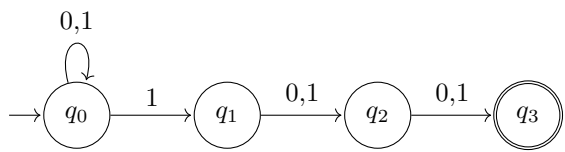
\begin{center}
  \begin{tikzpicture}[shorten >=1pt,node distance=2cm,on grid,auto] 
   \node[initial,state,initial text=] (q_0) {$q_0$}; 
   \node[state] (q_1) [right=of q_0] {$q_1$}; 
   \node[state] (q_2) [right=of q_1] {$q_2$};
   \node[state,accepting] (q_3) [right=of q_2] {$q_3$};
   \path[->] 
    (q_0) edge  node {1} (q_1)
          edge [loop above] node {0,1} ()
    (q_1) edge  node {0,1} (q_2)
    (q_2) edge  node {0,1} (q_3);
  \end{tikzpicture}
  \end{center}
  Here, $Q=\s{q_0, q_1, q_2, q_3}$, $\Sigma=\s{0,1}$, the collection
  of initial states is $\In=\s{q_0}$, the collection of final states
  is $\Fin=\s{q_3}$, and we have, e.g.,
  $\delta(q_0,1)=\s{q_0,q_1}$.
\end{example}
  
An automaton $A=\automaton$ can be used to specify a language
$\lang(A)\subseteq \Sigma^*$. Intuitively, $\lang(A)$ is the
collection of all the words over $\Sigma$ that specify a well-formed
walk along the state graph of $A$. The following definition makes this
intuition more precise.

\begin{definition}
  \label{def:language}
  Let $A=\automaton$ be an automaton. Then $A$ \emph{accepts} a word
  $w=w_1\cdots w_m\in\Sigma^*$ if there exists a sequence of states
  $s_0,s_1,\ldots, s_m\in Q$ such that
  \begin{enumerate}
  \item $s_0\in\In$,
  \item $s_{j+1}\in\delta(s_i,w_{j+1})$ for $j\in\s{0,\ldots,m-1}$,
    and
  \item $s_m\in \Fin$.
  \end{enumerate}
  The set of words accepted by $A$ is called the language
  \emph{recognized} by $A$ and is denoted $\lang(A)$.
\end{definition}

\begin{example}
  \label{ex:autlang}
  The alphabet for the automaton $A$ given in \cref{ex:aut} is $\Sigma
  = \s{0,1}$. The language recognized by $A$ is
  $\lang(A)=\s{w\in\Sigma^* ~;~ \mbox{the third rightmost letter of $w$
      is $1$}}$.
\end{example}  

If a language is recognized by some nondeterministic finite automata
then that language is called \emph{regular}. The collection of regular
languages is closed under a variety of operations. In particular,
regular languages are closed under concatenation.

\begin{definition}
  \label{def:comp}
  Let $A=\automaton$ and $A'=(\Sigma,Q',\In', \Fin',\delta')$ be two
  automata. Then the \emph{concatenation} of $A$ and $A'$ is the
  automaton $A\circ A'=(\Sigma,Q'',\In,\Fin',\delta'')$ where $Q'' =
  Q\sqcup Q'$ is the disjoint union of $Q$ and $Q'$ and
  \[
  \delta''(q,s) =
  \begin{cases}
    \delta(q,s)              & q\in Q \setminus \Fin, \\
    \delta(q,s)              & q\in \Fin \mbox{ and } s\neq \varepsilon, \\
    \delta(q,s)\cup \In'  & q\in \Fin \mbox{ and } s = \varepsilon, \mbox{ and} \\
    \delta'(q,s)             & q\in Q'.
  \end{cases}
  \]
\end{definition}

\begin{proposition}
  \label{prop:comp}
  Let $A$ and $A'$ be automata recognizing languages $L$ and $L'$,
  respectively. Then $A\circ A'$ recognizes $L\circ L'$.
\end{proposition}

An example of the concatenation of two automata is provided in
\cref{fig:autsc,ex:autsc} based off of the automata defined in
\cref{def:clifauto,def:gensauton} below.

\subsection{The Structure of Normal Forms}
\label{ssec:nfstructure}

We now consider the alphabet $\gens\cup\clifford$ and describe the
words over $\gens\cup\clifford$ that are output by the synthesis
algorithm of \cref{thm:nf}.

\begin{definition}
  Let $U\in\cliffordcs$. The \emph{normal form of $U$} is the unique
  word over $\gens\cup\clifford$ output by the synthesis algorithm of
  \cref{thm:nf} on input $U$. We write $\mathcal{N}$ for the
  collection of all normal forms.
\end{definition}

To describe the elements of $\mathcal{N}$, we introduce several
automata. It will be convenient for our purposes to enumerate the
elements of $\clifford$. We therefore assume that a total ordering of
the $92160$ elements of $\clifford$ is chosen and we write
$\clifford_j$ for the $j$-th element of $\clifford$.

\begin{definition}
  \label{def:clifauto}
  Let $k=|\clifford|$ and $\Sigma=\gens\cup\clifford$. The automaton
  $\mathfrak{C}$ is defined as $\mathfrak{C}=(\Sigma, [0,k], \s{0},
  [k],\delta_\mathfrak{C})$ where, for $s\in[0,k]$ and $\ell\in
  \Sigma$, we have
  \[
  \delta_\mathfrak{C}(s,\ell)=
  \begin{cases}
    \s{j} & \mbox{ if $s=0$ and $\ell = \clifford_j$, and}\\
    \varnothing & \mbox{ otherwise.}
  \end{cases}
  \]
\end{definition}

\begin{definition}
  \label{def:gensauton}
  Let $\Sigma=\gens\cup\clifford$. The automaton $\mathfrak{S}_{n,m}$
  is defined as $\mathfrak{S}_{n,m} = (\Sigma, [m], [n,m], [m],
  \delta_{\mathfrak{S},m})$ where, for $s\in[m]$ and $\ell\in \Sigma$,
  we have
  \[
  \delta_{\mathfrak{S},m}(s,\ell)=
  \begin{cases}
    \s{t ~;~ \partition(\oline{\gens_s})\cap
      \partition(\oline{\gens_t})=\varnothing} & \mbox{ if
      $\ell=\gens_s$ and}\\ \varnothing & \mbox{ otherwise.}
  \end{cases}
  \]
\end{definition}

\begin{example}
  \label{ex:autsc}
  To illustrate \cref{def:comp,def:clifauto,def:gensauton}, the
  automaton $\mathfrak{S}_{1,3}\circ\mathfrak{C}$ is represented in
  \cref{fig:autsc}. It can be verified that the words $\clifford_2$,
  $\gens_2\gens_1\clifford_1$, and $\gens_3\gens_1\gens_2\clifford_k$
  are accepted by $\mathfrak{S}_{1,3}\circ\mathfrak{C}$ while the
  words $\gens_1\gens_1\clifford_{4}$ and
  $\gens_3\clifford_{7}\gens_1$ are not. Note in particular that if
  $\clifford_1$ is the symbol for the identity, then
  $\gens_3\clifford_1$ is distinct (as a word) from $\gens_3$. The
  former is accepted by $\mathfrak{S}_{1,3}\circ\mathfrak{C}$ while
  the latter is not. Despite the state graph of $\mathfrak{S}_{1,3}$
  being fully-connected, full-connectivity does not necessarily hold
  for state graphs of other $\mathfrak{S}_{n,m}$ automata.
\end{example}

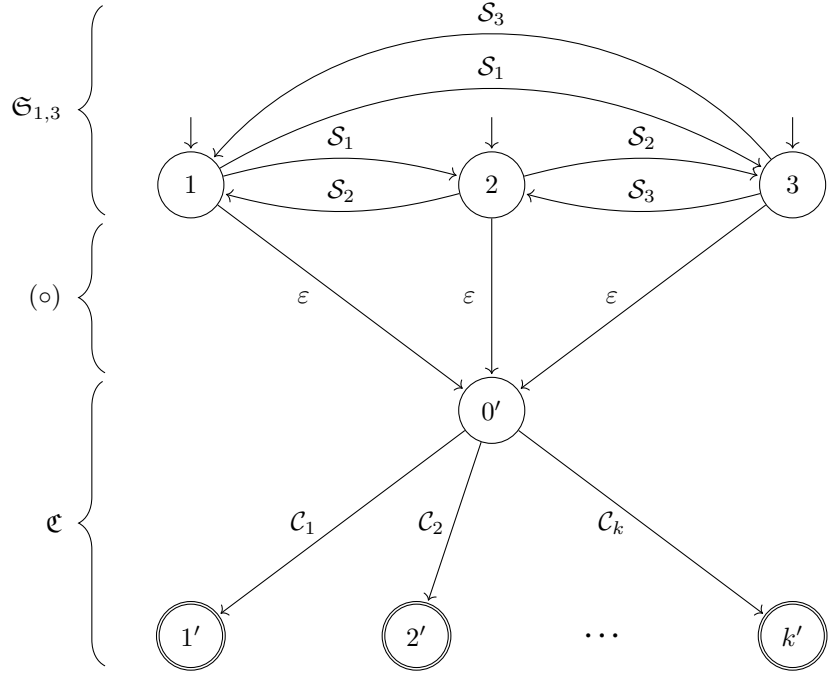
\begin{figure}
  \centering
  \begin{tikzpicture}[shorten >=1pt,node distance=3cm,on grid,auto]
  \draw [decorate,decoration={brace,amplitude=10pt},xshift=-4pt,yshift=0pt] (-1,-0.4) -- (-1,2.4) node [black,midway,xshift=-0.45cm] {$\mathfrak{S}_{1,3}$}; 
  \draw [decorate,decoration={brace,amplitude=10pt},xshift=-4pt,yshift=0pt] (-1,-2.5) -- (-1,-0.5) node [black,midway,xshift=-0.45cm] {$(\circ)$}; 
  \draw [decorate,decoration={brace,amplitude=10pt},xshift=-4pt,yshift=0pt] (-1,-6.4) -- (-1,-2.6) node [black,midway,xshift=-0.45cm] {$\mathfrak{C}$};
  \node[state,initial above,initial text=] at (0,0) (q_1) {$1$};
  \node[state,initial above,initial text=] at (4,0) (q_2) {$2$};
  \node[state,initial above,initial text=] at (8,0) (q_3) {$3$};
  \node[state,accepting] at (0,-6) (s_1) {$1'$};
  \node[state,accepting] at (3,-6) (s_2) {$2'$};
  \node[state,accepting] at (8,-6) (s_3) {$k'$};
  \node[state] at (4,-3) (t_2) {$0'$};
  \node[] at (5.5,-6) (d) {$\boldsymbol{\cdots}$};
  \node[] at (1.5,-4.5) (d) {$\clifford_1$};
  \node[] at (3.2,-4.5) (d) {$\clifford_2$};
  \node[] at (5.6,-4.5) (d) {$\clifford_k$};
  \node[] at (1.5,-1.5) (d) {$\varepsilon$};
  \node[] at (3.7,-1.5) (d) {$\varepsilon$};
  \node[] at (5.6,-1.5) (d) {$\varepsilon$};
  \path[->]
  (t_2) edge node {} (s_1)
  (t_2) edge node {} (s_2)
  (t_2) edge node {} (s_3)
  (q_1) edge node {} (t_2)
  (q_2) edge node {} (t_2)
  (q_3) edge node {} (t_2)
  (q_1) edge [bend left=15] node [above] {$\gens_1$} (q_2)
  (q_2) edge [bend left=15] node [above] {$\gens_2$} (q_1)
  (q_2) edge [bend left=15] node [above] {$\gens_2$} (q_3)
  (q_3) edge [bend left=15] node [above] {$\gens_3$} (q_2)
  (q_1) edge [bend left=30] node [above] {$\gens_1$} (q_3)
  (q_3) edge [bend right=50] node [above] {$\gens_3$} (q_1);
  \end{tikzpicture}
  \caption{The automaton $\mathfrak{S}_{1,3}\circ\mathfrak{C}$. The
    set of states of this automata is $\s{1,2,3,0',1',\dots,k'}$,
    which is the disjoint unioin of the states $\s{1,2,3}$ of
    $\mathfrak{S}_{1,3}$ and the states $\s{0,1,\dots,k}$ of
    $\mathfrak{C}$. The inital states are $\s{1,2,3}$, those of
    $\mathfrak{S}_{1,3}$, and the final states are $\s{1',\dots,k'}$,
    those of $\mathfrak{C}$. Because $\mathfrak{S}_{1,3}$ has $\Fin =
    \s{1,2,3}$ and $\mathfrak{C}$ has $\In=\s{0'}$, the transition
    function $\delta$ of $\mathfrak{S}_{1,3}\circ\mathfrak{C}$ is such
    that $\delta(1,\varepsilon) = \delta(2,\varepsilon) =
    \delta(3,\varepsilon) = \s{0'}$. Otherwise, $\delta$ behaves like
    the transition function for $\mathfrak{S}_{1,3}$ on the subset of
    states $\s{1,2,3}$ and like the transition function for
    $\mathfrak{C}$ on the subset of states
    $\s{0',1',\dots,k'}$.\label{fig:autsc}}
\end{figure}

We will use the automata introduced in
\cref{def:clifauto,def:gensauton} to describe the elements of
$\mathcal{N}$. Our goal is to show that
\begin{equation}
  \label{eq:nfaut}
  \mathcal{N}=\lang(\mathfrak{S}_{1,3} \circ \mathfrak{S}_{4,9} \circ
  \mathfrak{S}_{10,15} \circ \mathfrak{C})
\end{equation}
We start by establishing a few propositions.

\begin{proposition}
  \label{prop:inclusions}
  We have $\lang(\mathfrak{C})\subsetneq
  \lang(\mathfrak{S}_{1,15}\circ \mathfrak{C}) \subsetneq
  \lang(\mathfrak{S}_{1,9}\circ\mathfrak{S}_{10,15} \circ\mathfrak{C})
  \subsetneq \lang(\mathfrak{S}_{1,3} \circ \mathfrak{S}_{4,9} \circ
  \mathfrak{S}_{10,15} \circ \mathfrak{C})$, where $\subsetneq$
  denotes strict inclusion.
\end{proposition}

\begin{proof}
  By \cref{def:clifauto,def:gensauton}.
\end{proof}

We emphasize that the inclusions in \cref{prop:inclusions} are
strict. This implies that $\lang(\mathfrak{S}_{1,3} \circ
\mathfrak{S}_{4,9} \circ \mathfrak{S}_{10,15} \circ \mathfrak{C})$ can
be written as the disjoint union of $\lang(\mathfrak{C})$,
$\lang(\mathfrak{S}_{1,15}\circ \mathfrak{C})$, and
$\lang(\mathfrak{S}_{1,9}\circ\mathfrak{S}_{10,15}
\circ\mathfrak{C})$. The lemmas below show that these languages
correspond to disjoint subsets of $\mathcal{N}$ and, in combination,
suffice to prove \cref{eq:nfaut}.

\begin{lemma}
  \label{lem:clifauto}
  Let $U$ be a word over $\gens\cup\clifford$. Then
  $U\in\lang(\mathfrak{C})$ if and only if $U\in\mathcal{N}$ and
  $U$ has length $1$, i.e $U\in\clifford$.
\end{lemma}

\begin{proof}
  By \cref{def:clifauto,thm:nf}.
\end{proof}  

\begin{lemma}
  \label{lem:222pattern}
  Let $U$ be a word over $\gens\cup\clifford$. Then
  $U\in\lang(\mathfrak{S}_{1,15}\circ
  \mathfrak{C})\setminus\lang(\mathfrak{C})$ if and only if
  $U\in\mathcal{N}$ and $U$ has a $2\times2\times2$ pattern.
\end{lemma}

\begin{proof}
  First, note that $\lang(\mathfrak{C})$ is the set of words of length
  1 accepted by $\mathfrak{S}_{1,15}\circ \mathfrak{C}$. This means
  that $\lang(\mathfrak{S}_{1,15}\circ
  \mathfrak{C})\setminus\lang(\mathfrak{C})$ consists of all the words
  of length $k\geq 2$ accepted by $\mathfrak{S}_{1,15}\circ
  \mathfrak{C}$. Furthermore, by \cref{lem:CliffinD}, there are no
  normal forms of length 1 which have a $2\times2\times2$
  pattern. Thus, to prove our lemma it suffices to establish the
  following equality of sets
  \begin{align}
    \label{eq:222reword}
    \s{U\in\lang(\mathfrak{S}_{1,15}\circ \mathfrak{C})~;~|U|=k} =
    \s{U\in\mathcal{N}~;~ |U|=k \mbox{ and }\partition(U)\mbox{ is a
      }2\times 2\times 2\mbox{ pattern}}
  \end{align}
  for all $k\geq2$. We proceed by induction on $k$.
  \begin{itemize}
    \item Note that, by definition of $\mathfrak{S}_{1,15}\circ
      \mathfrak{C}$, we have $\s{U\in\lang(\mathfrak{S}_{1,15}\circ
        \mathfrak{C})~;~|U|=2} = \gens\clifford$. Every element of
      $\gens\clifford$ has a $2\times2\times 2$ pattern by
      \cref{lem:GensinD}. Moreover, for $U=SC$ with $S\in\gens$ and
      $C\in\clifford$, $\partition(SC) = \partition(S)$. Thus, $SC$
      must also be the unique word produced by the synthesis algorithm
      on input $U$ and hence $U\in\mathcal{N}$. This accounts for all
      words of length 2 in $\mathcal{N}$. Therefore
      \cref{eq:222reword} holds when $k=2$.
    \item Now suppose that \cref{eq:222reword} holds for some $k\geq
      2$. Let $U\in\lang(\mathfrak{S}_{1,15}\circ \mathfrak{C})$ be a
      word of length $k$ whose first letter is $S\in\gens$. Then
      $U\in\mathcal{N}$ and $\partition(U)=\partition(S)$ is a
      $2\times2\times2$ pattern. Furthermore, the least denominator
      exponent of $\oline{U}$ is $k-1$. We will show that
      \cref{eq:222reword} holds for $k+1$ by establishing two
      inclusions. Because it will sometimes be convenient to refer to
      submatrices, if $M$ is an $n\times n$ matrix and $x,y\subseteq
      [n]$, we write
      \[
      M[x;y]
      \]
      for the submatrix of $M$ formed from the rows with index in $x$
      and the columns with index in $y$.
      \begin{itemize}
        \item[$\subseteq$:] Suppose that $U'=S'U$ is a word of length
          $k+1$ accepted by $\lang(\mathfrak{S}_{1,15}\circ
          \mathfrak{C})$. Then by \cref{def:gensauton} we have
          $\partition(S')\cap\partition(S)=\varnothing$. Let
          $\s{a,b}\in\partition(S')$, and let $r_a$ and $r_b$ be the
          corresponding rows of the residue matrix of
          $\oline{U}$. Explicitly, we have
	  \begin{align*}
	    \rho_{k-1}(\oline{U})[\s{a,b};[6]]&=\begin{bmatrix}
	      r_a \\
	      r_b
	    \end{bmatrix}
	  \end{align*}
	  with $r_a\neq r_b$ as $\s{a,b}$ is not a subset of any
          element of $\partition(U)$. Direct calculation of the rows
          of the residue matrix for $\oline{U}'$ yields
	  \begin{align*}
	    \rho_{k}( \oline{U}')[\s{a,b};[6]]=\begin{bmatrix}
	      r_a + r_b \\
	      r_a + r_b
	    \end{bmatrix}.
	  \end{align*}
	  We conclude that $\s{a,b}$ is a subset of an element of
          $\partition(U')$. Furthermore, by
          \cref{lem:zerorows,eq:k2denom} we see that, since
          $r_a+r_b\neq 0$, $\partition(U')$ cannot be a $2\times4$
          pattern, and therefore $\s{a,b}\in\partition(U')$. As this
          holds for all $\s{a,b}\in\partition(S')$, we conclude that
          $\partition(S')=\partition(U')$. Thus, by the induction
          hypothesis, $S'U$ will be the word produced by the synthesis
          algorithm when applied to $U'$. Hence, $U'\in\mathcal{N}$
          and $\partition(U')$ is a $2\times2\times2$ pattern.
        \item[$\supseteq$:] Suppose that $U'$ is a normal form of
          length $k+1$ with a $2\times 2\times 2$ pattern. Write $U'$
          as $U'=S'V$ for some unknown normal form $V$. We then have
          $\partition(S')=\partition(U')$. Let
          $\s{a,b}\in\partition(S')$ and let the corresponding rows of
          the residue matrix of $\oline{V}$ be $r_a$ and
          $r_b$. Explicitly, we have
          \begin{align*}
            \rho_{k-1}(\oline{V})[\s{a,b};[6]]&=\begin{bmatrix}
              r_a \\
              r_b
            \end{bmatrix}.
          \end{align*}
	  Direct calculation of the rows of the residue matrix for
          $\oline{U}'$ yields
	  \begin{align*}
	    \rho_{k}( \oline{U}')[\s{a,b};[6]]=\begin{bmatrix}
	      r_a + r_b \\
	      r_a + r_b
	    \end{bmatrix}.
          \end{align*}
	  Since $\partition(U')$ is not a $2\times4$ pattern, we
          conclude that $r_a+r_b\neq0$ and thus that $r_a\neq
          r_b$. Therefore, there is no element of cardinality four in
          $\partition(V)$. Since $\lde(V)>0$, $\partition(V)$ must
          then be a $2\times2\times2$ pattern. Consequently, we have
          $V=U$ as defined above. Because
          $\s{a,b}\not\in\partition(U)=\partition(S)$, we know
          $\partition(S')\cap\partition(S)=\varnothing$. Given that
          $S'=\gens_{j'}$ and $S=\gens_j$, we conclude that
          $j\in\delta_{\mathfrak{S},15}(j',S'=\gens_{j'})$. Because
          $S=\gens_j$ is the first letter of the word $U$, we know the
          initial state of $U$ must be $j$. Therefore, by the
          induction hypothesis, $U'=S'U$ is accepted by
          $\mathfrak{S}_{1,15}\circ\mathfrak{C}$.
	\end{itemize}
      We have shown that \cref{eq:222reword} holds for words of length
      $k+1$ if it holds for words of length $k$. This completes the
      inductive step.\qedhere
  \end{itemize}
\end{proof}

\cref{lem:222pattern} characterized the normal forms that have a
$2\times 2 \times 2$ pattern. The two lemmas below jointly
characterize the normal forms that have a $2\times 4$ pattern. Because
their proofs are similar in spirit to that of \cref{lem:222pattern},
they have been relegated to \cref{app:walks}.

\begin{lemma}
  \label{lem:24pattern36}
  Let $U$ be a word over $\gens\cup\clifford$. Then $U\in
  \lang(\mathfrak{S}_{1,9}\circ\mathfrak{S}_{10,15}
  \circ\mathfrak{C})\setminus \lang(\mathfrak{S}_{1,15}\circ
  \mathfrak{C})$ if and only if $U\in \mathcal{N}$ and $U$ has a
  $2\times4$ pattern with 
  $\partition(U)\cap\s{\s{x,y}~;~(x,y)\in [3]\times[4,6]}\neq\varnothing$.
\end{lemma}

\begin{lemma}
  \label{lem:24pattern3366}
  Let $U$ be a word over $\gens\cup\clifford$. Then $U\in
  \lang(\mathfrak{S}_{1,3}\circ\mathfrak{S}_{4,9}\circ\mathfrak{S}_{10,15}\circ
  \mathfrak{C})\setminus
  \lang(\mathfrak{S}_{1,9}\circ\mathfrak{S}_{10,15}\circ
  \mathfrak{C})$ if and only if $U\in\mathcal{N}$ and $U$ has a
  $2\times4$ pattern with
  $\partition(U)\cap\s{\s{x,y}~;~(x,y)\in([3],[4,6])}=\varnothing$.
\end{lemma}

\begin{theorem}
  \label{thm:nfstructure}
  Let $U$ be a word over $\gens\cup\clifford$. Then $U\in
  \lang(\mathfrak{S}_{1,3} \circ \mathfrak{S}_{4,9} \circ
  \mathfrak{S}_{10,15} \circ \mathfrak{C})$ if and only if $U\in
  \mathcal{N}$.
\end{theorem}

\begin{proof}
  If $|U|=1$ then the result follows from \cref{lem:clifauto}. If
  $|U|>1$, then $U$ has a $2\times 2 \times 2$ or a $2\times 4$
  pattern and the result follows from \cref{prop:inclusions} and
  \cref{lem:222pattern,lem:24pattern36,lem:24pattern3366}.
\end{proof}

\section{Lower Bounds}
\label{sec:lowerbounds}

Recall that the \emph{distance} between operators $U$ and $V$ is
defined as $\lVert U-V\rVert = \sup\s{\lVert Uv - Vv \rVert ~;~
  \lVert v \rVert = 1}$. Because $\cliffordcs$ is universal, for every
$\epsilon >0$ and every element $U\in\su(4)$, there exists
$V\in\cliffordcs$ such that $\lVert U-V \rVert \leq \epsilon$. In such
a case we say that $V$ is an \emph{$\epsilon$-approximation} of
$U$. We now take advantage of \cref{thm:nfstructure} to count
Clifford+$CS$ operators and use these results to derive a worst-case
lower bound on the $CS$-count of approximations.

\begin{lemma}
  \label{lem:exactlyn}
  Let $n\geq 1$. There are $86400(3\cdot 8^n-2\cdot 4^n)$
  Clifford+$CS$ operators of $CS$-count exactly $n$.
\end{lemma}

\begin{proof}
  Each Clifford+$CS$ operator is represented by a unique normal form
  and this representation is $CS$-optimal. Hence, to count the number
  of Clifford+$CS$ operators of $CS$-count $n$, it suffices to count
  the normal forms of $CS$-count $n$. By \cref{thm:nfstructure}, and
  since Clifford operators have $CS$-count 0, a normal form of
  $CS$-count $n$ is a word
  \begin{equation}
    \label{eq:words}
    w = w_1w_2w_3w_4
  \end{equation}
  such that $w_1\in\lang(\mathfrak{S}_{1,3})$,
  $w_2\in\lang(\mathfrak{S}_{4,9})$,
  $w_3\in\lang(\mathfrak{S}_{10,15})$, $w_4\in\lang(\mathfrak{C})$ and
  the $CS$-counts of $w_1$, $w_2$, and $w_3$ sum to $n$. There are
  \begin{equation}
    \label{eq:words1}
    (6\cdot 8^{n-1}+6\cdot 4^{n-1}+3\cdot 2^{n-1})\cdot |\clifford|
  \end{equation}
  words of the form of \cref{eq:words} such that exactly one of $w_1$,
  $w_2$, or $w_3$ is not $\varepsilon$. Similarly, there are
  \begin{equation}
    \label{eq:words2}
    \left(\sum_{0<l<n}18\cdot 2^{2n-3-l}+\sum_{0<l<n}18\cdot 2^{3n-4-2l}
    +\sum_{0<j<n}36\cdot 2^{3n-5-j} \right)\cdot |\clifford|
  \end{equation}
  words of the form of \cref{eq:words} such that exactly two of $w_1$,
  $w_2$, or $w_3$ are not $\varepsilon$. Finally, the number of words
  of the form of \cref{eq:words} such that $w_1$, $w_2$, and
  $w_3$ are not $\varepsilon$ is 
  \begin{equation}
    \label{eq:words3}
    \left(\sum_{0<l<n-j} \sum_{0<j<n}108\cdot 2^{3n-6-j-2l}\right)\cdot |\clifford|
  \end{equation}
  Summing \cref{eq:words1,eq:words2,eq:words3} and applying the geometric
  series formula then yields the desired result.
\end{proof}

\begin{corollary}
  \label{cor:upton}
  For $n\in\N$, there are $\frac{46080}{7}(45\cdot 8^n-35\cdot 4^n+4)$
  distinct Clifford+$CS$ operators of $CS$-count at most $n$.
\end{corollary}

\begin{proof}
  Recall that the Clifford+$CS$ operators of $CS$-count $0$ are
  exactly the Clifford operators and that $|\clifford|=92160$. The
  result then follows from \cref{lem:exactlyn} and the geometric
  series formula.
\end{proof}

\begin{proposition}
  For every $\epsilon\in\R^{>0}$, there exists $U\in \su(4)$ such that
  any Clifford+$CS$ $\epsilon$-approximation of $U$ has $CS$-count at
  least $5\log_2 (1/\epsilon) -0.67$.
\end{proposition}

\begin{proof}
  By a volume counting argument. Each operator must occupy an
  $\epsilon$-ball worth of volume in 15-dimensional $\su(4)$ space,
  and the sum of all these volumes must add to the total volume of
  $\su(4)$ which is $(\sqrt{2}\pi^9)/3$. The number of circuits up to
  $CS$-count $n$ is taken from \cref{cor:upton} (we must divide the
  result by two to account for the absence of overall phase $\omega$
  in the special unitary group) and a 15-dimensional $\epsilon$-ball
  has a volume of
  \[
  \frac{\pi^{\frac{15}{2}}}{\Gamma \left(\frac{15}{2}+1\right)}
  \epsilon^{15}. \qedhere
  \]
\end{proof}

Let $U$ be an element of $\cliffordcs$ of determinant 1. By
\cref{eq:u4rep} of \cref{sec:gens}, $U$ can be written as
\[
  U=\frac{1}{\sqrt{2}^k} M
\]
where $k\in\N$ and the entries of $M$ belong to $\Zi$. We can
therefore talk about the least denominator exponent of the $\su(4)$
representation of $U$. We finish this section by relating the least
denominator exponent of the $\su(4)$ representation of $U$ and the
$CS$-count of the normal form of $U$.

\begin{proposition}
  \label{prop:ldecount}
  Let $U$ be an element of $\cliffordcs$ of determinant 1, let $k$ be
  the least denominator exponent of the $\su(4)$ representation of
  $U$, and let $k'$ be the $CS$-count of the normal form of $U$. Then
  \[
  \frac{k-3}{2}\leq k' \leq 2k+2.
  \]
\end{proposition}

\begin{proof}
  The $CS$-count of the normal form of $U$ is equal to the least
  denominator exponent of the $\so(6)$ representation of
  $U$. \cref{eq:rep} then implies the upper bound for $k'$.  Likewise,
  examination of \cref{thm:nfstructure} reveals that the $CS$
  operators in the circuit for $U$ must be separated from one another
  by a Clifford with a least denominator exponent of at most 2 in its
  unitary representation. Combining this with the fact that the
  largest least denominator exponent of an operator in $\clifford$ is
  3, we arrive at the lower bound for $k'$.
\end{proof}

\begin{remark}
  It was established in \cite{RS16} that, for single-qubit
  Clifford+$T$ operators of determinant $1$, there is a simple
  relation between the least denominator exponent of an operator and
  its $T$-count: if the least denominator exponent of the operator is
  $k$, then its $T$-count is $2k-2$ or $2k$. Interestingly, this is
  not the case for Clifford+$CS$ operators in $\su(4)$, as suggested
  by \cref{prop:ldecount}. Clearly, the $CS$-count of an operator
  always scales linearly with the least denominator exponent of its
  unitary representation. For large $k$, computational experiments
  with our code \cite{thecode} suggest that most operators are such
  that $k'\approx k$, though there are examples of operators with
  $k'\approx 2k$. One example of such an operator is $\left[R(X
    \otimes I,I \otimes Z)R(X \otimes I,I \otimes X)R(Z \otimes I,I
    \otimes X)R(Z \otimes I,I \otimes Z)\right]^m$ for $m\in\N$.
\end{remark}

\section{Conclusion}
\label{sec:conc}

We described an exact synthesis algorithm for a fault-tolerant
multi-qubit gate set which is simultaneously optimal, practically
efficient, and explicitly characterizes all possible outputs. The
algorithm establishes the existence of a unique normal form for
two-qubit Clifford+$CS$ circuits. We showed that the normal form for
an operator can be computed with a number of arithmetic operations
linear in the gate-count of the output circuit. Finally, we used a
volume counting argument to show that, in the typical case,
$\epsilon$-approximations of two-qubit unitaries will require a
$CS$-count of at least $5\log_2(1/\epsilon)$.

We hope that the techniques developed in the present work can be used
to obtain optimal multi-qubit normal forms for other two-qubit gate
sets, such as the two-qubit Clifford+$T$ gate set. Indeed, it can be
shown that the $\so(6)$ representation of Clifford+$T$ operators are
exactly the set of $\so(6)$ matrices with entries in the ring
$\Z[1/\sqrt{2}]$. Further afield, the exceptional isomorphism for
$\su(8)$ could potentially be leveraged to design good synthesis
algorithms for three-qubit operators. Such algorithms would provide a
powerful basis for more general quantum compilers.

An interesting avenue for future research is to investigate whether
the techniques and results presented in this paper can be used in the
context of \emph{synthillation}. Quantum circuit synthesis and magic
state distillation are often kept separate. But it was shown in
\cite{CH161} that performing synthesis and distillation simultaneously
(synthillation) can lead to overall savings. The analysis presented in
\cite{CH161} uses $T$ gates and $T$ states. Leveraging
higher-dimensional synthesis methods such as the ones presented here,
along with distillation of $CS$ states, could yield further savings.

\section{Acknowledgements}
\label{sec:acknowledgements}

AG was partially supported by the Princeton Center for Complex
Materials, a MRSEC supported by NSF grant DMR 1420541. NJR was
partially supported by the Natural Sciences and Engineering Research
Council of Canada (NSERC), funding reference number RGPIN-2018-04064.

We would like to thank Matthew Amy, Xiaoning Bian, and Peter Selinger
for helpful discussions. In addition, we would like to thank the
anonymous reviewers whose comments greatly improved the paper.

\section{Contributions}
\label{sec:contributions}

All authors researched, collated, and wrote this paper.

\section{Competing Interests}
\label{sec:compinterests}

The authors declare no competing interests.

\section{Data Availability}
\label{sec:dataavail}

The sets of various $CS$-count operators used to generate the
algorithmic performance information in \cref{table:performance} are
available at \cite{thecode}.

\bibliographystyle{abbrv} 
\bibliography{cliffordcs}

\appendix

\section{Computing SO(6) Representations}
\label{app:calculation}

	Consider the unitary matrix
	\[
		U = \frac{1}{4}\begin{bmatrix}
			3+i & -1-i & -2 & 0\\
			1-i & 3-i & 0 & -2 \\
			2 & 0 & 3-i & 1+i\\
			0 & 2 & -1+i & 3+i
		\end{bmatrix}.
	\]	
	Suppose we want to compute the entry in the third row and
        fourth column of it's equivalent SO(6) representative
        $\overline{U}$, i.e. $\overline{U}_{3,4}$.Note that we have
        $\det(U) = 1$, so $U\in\su(4)$ and we do not have to multiply
        by a phase before mapping $U$ to $\overline{U}$.
	
	We need
	\begin{align*}
		B_3 &= s_{-,23,14} = \frac{i}{\sqrt{2}}(e_2\wedge e_3-e_1\wedge e_4)\\
		B_4 &= s_{+,24,13} = \frac{1}{\sqrt{2}}(e_2\wedge e_4+e_3\wedge e_1)
	\end{align*}
	in order to calculate
	\[
		\overline{U}_{3,4} = \langle B_3,\overline{U} B_4\rangle.
	\]
	Computing $\overline{U}B_4$ directly, we have
	\begin{align*}
		\overline{U} B_4 &= \frac{1}{\sqrt{2}}\left[ \overline{U}(e_2\wedge e_4)+\overline{U}(e_3\wedge e_1)\right]\\
		&= \frac{1}{\sqrt{2}}\left[ (U e_2)\wedge (U e_4)+(U e_3)\wedge (U e_1)\right]\\
		&= \frac{1}{16\sqrt{2}}\left\{ \left[(-1-i)e_1+(3-i)e_2+2 e_4\right]\wedge \left[-2 e_2 + (1+i) e_3 + (3+i) e_4 \right]\right.\\
		&\qquad\qquad \left.\left[-2 e_1 + (3-i) e_3 + (-1+i) e_4 \right]\wedge\left[(3+i)e_1 + (1-i)e_2  + 2 e_3 \right]\right\}\\
		&= \frac{1}{16\sqrt{2}}\left[ (2+2i) (e_1\wedge e_2) - 2i (e_1\wedge e_3) - (2+4i) (e_1\wedge e_4) + (-6+2i)(e_2\wedge e_2)\right.\\
		&\qquad\qquad \left.+ (4+2i)(e_2\wedge e_3) +10(e_2\wedge e_4) -4(e_4\wedge e_2) + (2+2i) (e_4\wedge e_3)\right.\\
		&\qquad\qquad\left. + (6+2i)(e_4\wedge e_4) -(6+2i)(e_1\wedge e_1) +(-2+2i) (e_1\wedge e_2)-4(e_1\wedge e_3)\right.\\
		&\qquad\qquad\left.  +10 (e_3\wedge e_1)+(2-4i)(e_3\wedge e_2)+ (6-2i)(e_3\wedge e_3) \right.\\
		&\qquad\qquad\left. +(-4+2i)(e_4\wedge e_1)+2i(e_4\wedge e_2)  + (-2+2i)(e_4\wedge e_3)\right].
	\end{align*}
	Using the anticommutation of the wedge product, we can
        simplify this expression as
	\begin{align*}
		\overline{U} B_4 &= \frac{1}{8\sqrt{2}}\left[ 2i(e_1\wedge e_2) +(-7-i) (e_1\wedge e_3) + (1-3i)(e_1\wedge e_4)\right.\\
		&\qquad\qquad\quad\left. +(1+3i)(e_2\wedge e_3) + (7-i)(e_2\wedge e_4) -2i(e_3\wedge e_4)\right].
	\end{align*}
	Examining the rule for computing inner products of wedge products, we see that we have
	\[
		\langle e_i \wedge e_j, e_k \wedge e_\ell \rangle = \langle e_i, e_k \rangle \langle e_j, e_\ell\rangle - \langle e_i, e_\ell \rangle \langle e_j, e_k \rangle = \delta_{i,k} \delta_{j,\ell} - \delta_{i,\ell} \delta_{j,k}
	\]
	and using this property along with the linearity of the inner product, we compute
	\begin{align*}
		\overline{U}_{3,4} &= \langle B_3,\overline{U} B_4\rangle\\
		&=-\frac{i}{16}\left[2i\langle e_2\wedge e_3,e_1 \wedge e_2\rangle - 2i\langle e_1\wedge e_4,e_1 \wedge e_2\rangle + (-7-i)\langle e_2\wedge e_3,e_1 \wedge e_3\rangle\right.\\
		&\qquad\quad\left.-(-7-i)\langle e_1\wedge e_4,e_1 \wedge e_3\rangle + (1-3i)\langle e_2\wedge e_3,e_1 \wedge e_4\rangle - (1-3i)\langle e_1\wedge e_4,e_1 \wedge e_4\rangle\right.\\
		&\qquad\quad\left. +(1+3i)\langle e_2\wedge e_3,e_2\wedge e_3\rangle - (1+3i)\langle e_1\wedge e_4,e_2\wedge e_3\rangle +(7-i)\langle e_2\wedge e_3,e_2\wedge e_4\rangle\right.\\
		&\qquad\quad\left. -(7-i)\langle e_1\wedge e_4,e_2\wedge e_4\rangle -2i\langle e_2\wedge e_3,e_3\wedge e_4\rangle+2i\langle e_1\wedge e_4,e_3\wedge e_4\rangle\right]\\
		&= -\frac{i}{16}\left[-(1-3i)+(1+3i)\right]\\
		&=\frac{3}{8}.
	\end{align*}
	Computing the full matrix, we have
	\[
		\overline{U} = \frac{1}{8}\begin{bmatrix}
			4 & 0 & 6 & 2 & 2 & -2\\
			0 & 8 & 0 & 0 & 0 & 0\\
			-6 & 0 & 1 & 3 & 3 & -3\\
			2 & 0 & -6 & 7 & -1  & 1\\
			2 & 0 & -3 & -1 & 7 & 1\\
			-2 & 0 & 3 & 1 & 1 & 7
		\end{bmatrix}.
	\]

\section{Proof of \texorpdfstring{\cref{lem:rels}}{Lemma~2.2}}
\label{app:proof}

This appendix contains a proof of \cref{lem:rels}, whose statement we
reproduce below for completeness.

\begin{lemma*}
  Let $C\in\clifford$ and let $P$, $Q$, and $L$ be distinct elements
  of $\pauli \setminus \s{I}$. Assume that $P$, $Q$, and $L$ are
  Hermitian and that $PQ=QP$, $PL=LP$, and $QL=-LQ$. Then the
  following relations hold:
  \begin{align}
  C R(P,Q)C^\dagger & = R(CPC^\dagger,CQC^\dagger), \label{eq:CliffordCommuteapp}\\    
  R(P,Q) & = R(Q,P), \label{eq:swappableapp}\\  
  R(P,-PQ) & = R(P,Q), \label{eq:permutableapp}\\
  R(P,-Q) & \in R(P,Q) \clifford, \label{eq:minusPauliapp}\\
  R(P,Q)^2 & \in \clifford,\mbox{ and} \label{eq:squaredapp}\\
  R(P,L) R(P,Q) & = R(P,Q) R(P,iQL). \label{eq:sharedPauliapp}
  \end{align}
\end{lemma*}

\begin{proof}
  Since $\clifford$ is the normalizer of $\pauli$, $CPC^\dagger$ and
  $CQC^\dagger$ are Hermitian and commuting elements of
  $\pauli\setminus\s{I}$. \cref{eq:CliffordCommuteapp} then follows
  \cref{def:gens}. \cref{eq:swappableapp} is a direct consequence of
  \cref{def:gens}. Now, we have $R(P,Q) = \exp((i\pi/8) (\Id-P-Q+PQ))$
  and therefore
  \[
  R(P,-PQ) = \exp((i\pi/8)(\Id-P-(-PQ)+(P)(-PQ))) =
  \exp((i\pi/8)(\Id-P-Q+PQ)) = R(P,Q),
  \]
  which proves \cref{eq:permutableapp}. For \cref{eq:minusPauliapp},
  first note that
  	\begin{align}
  	R(P,-Q)&=\exp((i\pi/8)(\Id-P-(-Q)+(P)(-Q)))=R(P,Q)\nonumber\\
  	&=\exp((i\pi/8)(\Id-P-Q+PQ))\exp((i\pi/4)Q)\exp((i\pi/4)(-PQ))\nonumber\\
  	&= R(P,Q)\exp((i\pi/4)Q)\exp((i\pi/4)(-PQ)).\label{eq:minuspauli}
  	\end{align}
  We now show that the last two terms in \cref{eq:minuspauli} are
  Clifford operators. For any $A,B\in\pauli$, either $AB=BA$ or
  $AB=-BA$ holds. For Hermitian Pauli $A$ and any $B\in\pauli$, we
  compute
  \begin{align*}
  	\exp((i\pi/4)A) B \exp((i\pi/4)A)^\dagger &= \frac{I+iA}{\sqrt{2}} B \frac{I-iA}{\sqrt{2}}\\
  	&= A\frac{AB+BA}{2} + i \frac{AB -BA}{2}\\
  	&=\begin{cases}
  	B\in\pauli & \mbox{if}~AB = BA\\
  	iAB\in\pauli & \mbox{if}~AB = -BA
  	\end{cases}
  \end{align*}
  By the definition of the Clifford group and as
  $\exp((i\pi/4)I)=\exp(i\pi/4) I\in\clifford$, we therefore conclude
  \begin{align}
  \label{eq:pauliexplicity}
  	\exp((i\pi/4)A)\in\clifford ~\mbox{for all}~ A\in\pauli~\mbox{with}~A^\dagger=A.
  \end{align}
  As both $Q$ and $-PQ$ are Hermitain Paulis, we conclude from
  \cref{eq:minuspauli,eq:pauliexplicity} that
  \[
  	R(P,-Q) = R(P,Q)\exp((i\pi/4)Q)\exp((i\pi/4)(-PQ)) \in R(P,Q)\clifford
  \]
  which proves \cref{eq:minusPauliapp}. From \cref{eq:pauliexplicity}
  we also conclude \cref{eq:squaredapp} holds as $\s{I,-P,-Q,PQ}$ is a
  set of commuting Hermitian Paulis and thus
  \[
	R(P,Q)^2 = \exp((i\pi/4) \Id)\exp((i\pi/4) (-P))\exp((i\pi/4) (-Q))\exp((i\pi/4) PQ)\in\clifford.
  \]
  Finally, note that if $R(P,Q)$ is as in \cref{def:gens} then
  $(I-P)/2$ and $(I-Q)/2$ are idempotent. We can therefore explicitly
  compute the exponential to get
  \[
  R(P,Q) = \Id + (i-1)\left(\frac{\Id-P}{2}\right)
  \left(\frac{\Id-Q}{2}\right).
  \]
  Using the above expression, together with the fact that $QL=-LQ$,
  yields
  \begin{align*}
    R(P,L) R(P,Q) & = I+(i-1) \rbk{\frac{I-P}{2}}
    \sbk{\frac{3+i}{4}I-\frac{1+i}{4}L-\frac{1+i}{4}Q-\frac{1+i}{4}(-iLQ)}\\ &
    = I+(i-1) \rbk{\frac{I-P}{2}} \sbk{\frac{3+i}{4}I-
      \frac{1+i}{4}(-iQiQL) - \frac{1+i}{4}Q-\frac{1+i}{4}(iQL)} \\ &
    = I+(i-1) \rbk{\frac{I-P}{2}} \sbk{\frac{3+i}{4}I-
      \frac{1+i}{4}Q-\frac{1+i}{4}iQL - \frac{1+i}{4}(-iQiQL)} \\ & =
    R(P,Q)R(P,iQL),
  \end{align*}
  which proves \cref{eq:sharedPauliapp} and thereby completes our
  proof.
\end{proof}

\section{Proofs of \texorpdfstring{\cref{lem:24pattern36}}{Lemma~5.16} and
\texorpdfstring{\cref{lem:24pattern3366}}{Lemma~5.17}}
\label{app:walks}

In this appendix, we provide proofs for \cref{lem:24pattern36} and
\cref{lem:24pattern3366}. We first establish a useful lemma for
notational purposes.

\begin{lemma}
	\label{applem:reverseorder}
	Let $U = S_1 \cdots S_k$ be a word over $\gens\cup\clifford$ such that $U\in\lang(\mathfrak{S}_{1,15})$. Then $\partition(U^{-1}) = \partition(S_k)$.
\end{lemma}

\begin{proof}
	By \cref{lem:222pattern}, we know $\oline{U}$ has least denominator exponent $k$ and a $2\times2\times2$ pattern, and by \cref{lem:GensinD,lem:zerorows} we observe that $\oline{U}^\Trans$ must likewise have a $2\times2\times2$ pattern. Therefore, if $V$ is the normal form equivalent to $U^{-1}$, we have $V\in\lang(\mathfrak{S}_{1,15}\circ\mathfrak{C})$ such that $\oline{V} = \oline{U}^\Trans$ has least denominator exponent $k$ and $\partition(V) = \partition(U^{-1})$. By \cref{lem:rels}, we know that $\oline{S}_k^2\in\oline{\clifford}$ and so $U S_k$ has an $\so(6)$ representation with least denominator exponent $k-1$. Then we have that $\oline{S}_k^\Trans \oline{V}$ has least denominator exponent $k-1$, and as $V$ has a $2\times2\times2$ pattern, $S_k$ is the first letter of the normal form $V$. Thus we conclude that $\partition(S_k) = \partition(V) = \partition(U^{-1})$.
\end{proof}

We use \cref{applem:reverseorder} to concisely describe the pattern of the rightmost letter of a word $U\in\lang(\mathfrak{S}_{a,b})\subseteq\lang(\mathfrak{S}_{1,15})$ throughout the remaining lemmas.

\begin{lemma*}
	\label{applem:24pattern36}
	Let $U$ be a word over $\gens\cup\clifford$. Then $U\in
	\lang(\mathfrak{S}_{1,9}\circ\mathfrak{S}_{10,15}
	\circ\mathfrak{C})\setminus \lang(\mathfrak{S}_{1,15}\circ
	\mathfrak{C})$ if and only if $U\in \mathcal{N}$ and $U$ has a
	$2\times4$ pattern with 
	$\partition(U)\cap\s{\s{x,y}~;~(x,y)\in [3]\times[4,6]}\neq\varnothing$.
\end{lemma*}

\begin{proof}
	By \cref{lem:clifauto,lem:222pattern}, $\lang(\mathfrak{S}_{1,15}\circ\mathfrak{C})$ accounts for all Clifford and $2\times2\times2$ pattern normal forms. Any accepting word of $\mathfrak{S}_{1,9}\circ\mathfrak{S}_{10,15}\circ\mathfrak{C}$ which is not accepted by $\mathfrak{S}_{1,15}\circ\mathfrak{C}$ is of the form $U_1 U_2$ where $U_1\in\lang(\mathfrak{S}_{1,9})\subsetneq \lang(\mathfrak{S}_{1,15})$, $U_2\in\lang(\mathfrak{S}_{10,15}\circ\mathfrak{C})\subsetneq \lang(\mathfrak{S}_{1,15}\circ\mathfrak{C})$, and $\partition(U_1^{-1})\cap\partition(U_2)\neq \varnothing$. We can then restate our lemma as follows:
	\begin{gather}
		\s{U_1 U_2~;~ U_1\in\lang(\mathfrak{S}_{1,9}),\,U_2\in\lang(\mathfrak{S}_{10,15}\circ\mathfrak{C}),\,|U_1| = \ell,\,|U_2| = k - \ell,\, \mbox{and } \partition(U_1^{-1})\cap \partition(U_2)\neq\varnothing} \nonumber\\
		= \label{eq:24lemrestate}\\
		\s{U\in\mathcal{N}~;~ |U| = k \mbox{ and } \partition(U) \mbox { is a }2\times 4 \mbox{ pattern s.t. } \partition(U)\cap \s{\s{x,y}~;~(x,y)\in [3]\times[4,6]}\neq\varnothing}\nonumber
	\end{gather}
	for all $k\geq 3$ and $1\leq\ell\leq k-2$. We begin with a useful result, and afterwards proceed by induction on $k$.
	
	\begin{itemize}
		\item 
		Consider a length $k$ accepting word $U'$ of the above
                form such that $\ell=1$. Then $U' = S' U$ where, given
                that the first letter of $U$ is $S$, we have
                $S'\in\s{\gens_1,\cdots, \gens_9}$, $S\in
                \s{\gens_{10},\cdots, \gens_{15}}$, and
                $U\in\lang(\mathfrak{S}_{10,15}\circ\mathfrak{C})$. We
                know that we must have
                $\partition(S')\cap\partition(S)\neq\varnothing$, and
                by inspection the only way this is achieved is if
                $\partition(S')\cap \partition(S) =
                \s{\s{a,b}}\subseteq\s{\s{x,y}~;~(x,y)\in
                  [3]\times[4,6]}$. As $\partition(S'^{-1})=
                \partition(S')$ is not finer than
                $\partition(S)=\partition(U)$, we conclude that the
                least denominator exponent of $\oline{U}'$ is $k-1$ by
                \cref{lem:denomreduce,lem:222pattern}. Furthermore, we
                have
		\[
			\rho_{k-1} (\oline{U}')[\s{a,b};[6]] = \begin{bmatrix}
			0 \\ 0
			\end{bmatrix}
		\]
		and so $U'$ must have a $2\times 4$ pattern with $\s{a,b}\in\partition(U')$. This means that 
		\[
		\partition(U')\cap \s{\s{x,y}~;~(x,y)\in [3]\times[4,6]}\neq\varnothing
		\]
		and $\partition(S')$ is finer than $\partition(U')$. For each such $2\times 4$ pattern, there is one and only one element of $\s{\gens_1,\cdots, \gens_9}$ which is a finer partition than $\partition(U')$, and therefore $S'$ is the leftmost syllable of the normal form equivalient to $U'$ under FFP. As $U\in\mathcal{N}$, we thus conclude that $U'\in\mathcal{N}$.
		
		\item We now show that accepting words of the above form with $k=3$ (enforcing $\ell=1$) constitute all such length 3 normal forms with the described pattern. Clearly, any $U\in\mathcal{N}$ with the aforementioned pattern must be of the form $U=S_1 S_2 C$ where $S_1\in\s{\gens_1,\cdots \gens_9}$, $S_2\in\gens$, and $C\in\clifford$. Furthermore, we must have $\partition(S_1)\cap \partition(S_2) = \s{\s{a,b}}$ with $\s{a,b}\in\s{\s{x,y}~;~(x,y)\in [3]\times[4,6]}$ to produce the appropriate $2\times 4$ pattern. By inspection, this implies that $S_2\in\s{\gens_{10},\cdots,\gens_{15}}$, and so $U$ is accepted by $\mathfrak{S}_{1,9}\circ\mathfrak{S}_{10,15}\circ\mathfrak{C}$ but not $\mathfrak{S}_{1,15}\circ\mathfrak{C}$ as required. Therefore, \cref{eq:24lemrestate} holds for $k=3$.
		\item Now suppose that \cref{eq:24lemrestate} holds for some $k\geq
		3$. We will show that
		\cref{eq:24lemrestate} holds for $k+1$ by establishing two
		inclusions.
		\begin{itemize}
			\item[$\subseteq$:] We have already proven this inclusion in the case of $\ell=1$ for all $k$. We therefore need only consider the $\ell>1$ case. Let $U\in\lang(\mathfrak{S}_{1,9}\circ\mathfrak{S}_{10,15}\circ \mathfrak{C})\setminus\lang(\mathfrak{S}_{1,15}\circ\mathfrak{C})$ be a word of length $k$ whose first letter is $S\in\s{\gens_1,\cdots,\gens_9}$. Then
			$U\in\mathcal{N}$ and $\partition(U)$ is a
			$2\times4$ pattern such that $\partition(U)\cap\partition(S)=\s{\s{a,b}}$ with $a\in[3]$ and $b\in[4,6]$. By inspection, we explicitly have $S = \s{\s{a,b},\s{c,d},\s{e,f}}$ with $\s{a,c,d}=[3]$ and $\s{b,e,f}=[4,6]$. Furthermore, the least denominator exponent of $\oline{U}$ is $k-1$. Suppose that $U'=S'U$ is a word of length
			$k+1$ accepted by $\mathfrak{S}_{1,9}\circ\mathfrak{S}_{10,15}\circ
			\mathfrak{C}$. Then by \cref{def:gensauton} we have
			$\partition(S')\cap\partition(S)=\varnothing$ with $S'\in\s{\gens_1,\cdots,\gens_9}$. By inspection, this implies that $S'$ must have the pattern $\partition(S') = \s{\s{a,c},\s{b,e},\s{d,f}}$ (up to $c\leftrightarrow d$ and $e\leftrightarrow f$) with $\s{a,b,c,d,e,f}$ as defined before. As $\partition(S'^{-1}) = \partition(S')$ is such that $\partition(S'^{-1})\cap\partition(U)=\varnothing$, then $\partition(S'^{-1})$ is not finer than $\partition(U)$ and so the least denominator exponent of $U'$ must be $k$. Consider the set $\s{d,f}$ and let $r_d$ and $r_f$ be the corresponding rows of the residue matrix of $\oline{U}$. Explicitly, we have
			\[
				\rho_{k-1}(\oline{U})[\s{d,f};[6]] = \begin{bmatrix}
				r_d \\ r_f
				\end{bmatrix}
			\]
			with $r_d = r_f$ as $\s{d,f}$ is a subset of the cardinality four element of $\partition(U)$. Directly calculating the rows of the residue matrix for $\oline{U}'$ yields
			\[
				\rho_k(\oline{U}')[\s{d,f};[6]] = \begin{bmatrix}
				r_d + r_f \\ r_d + r_f
				\end{bmatrix} = 
				\begin{bmatrix}
				0 \\ 0
				\end{bmatrix}
			\]
			and we therefore conclude that $\s{d,f}\in\partition(U')$ and $U'$ has a $2\times4$ pattern. As $\partition(S')$ is the lowest-indexed element of $\gens$ finer than $\partition(U')$, under FFP we conclude that $S'$ is the leftmost syllable of the normal form equivalent to $U'$. Since $U\in\mathcal{N}$ by assumption, we therefore conclude $U' = S' U\in\mathcal{N}$, and have established that $U'$ has a $2\times4$ pattern such that $\partition(U')\cap\s{\s{x,y}~;~(x,y)\in [3]\times[4,6]}\neq\varnothing$.
			\item[$\supseteq$:]
			Suppose that $U'$ is a normal form of
			length $k+1$ with a $2\times 4$ pattern such that $\partition(U')\cap\s{\s{x,y}~;~(x,y)\in [3]\times[4,6]}\neq\varnothing$. Write $U'$
			as $U'=S'U$ for some unknown normal form $U$. We then have $S'\in\s{\gens_1,\cdots,\gens_9}$ and $\partition(S') = \s{\s{a,b},\s{c,d},\s{e,f}}$ such that $\s{a,c,d}=[3]$, $\s{b,e,f}=[4,6]$, and $\partition(S')\cap\partition(U')=\s{\s{a,b}}$. For
			$\s{i,j}\in\partition(S')$, let the corresponding rows of
			the residue matrix of $\oline{U}$ be $r_i$ and
			$r_j$. Explicitly, we have
			\begin{align*}
			\rho_{k-1}(\oline{U})[\s{i,j};[6]]&=\begin{bmatrix}
			r_i \\
			r_j
			\end{bmatrix}.
			\end{align*}
			Direct calculation of the rows for the residue matrix of
			$\oline{U}'$ yields
			\begin{align*}
			\rho_{k}( \oline{U}')[\s{i,j};[6]]=\begin{bmatrix}
			r_i + r_j \\
			r_i + r_j
			\end{bmatrix}.
			\end{align*}
			As $r_a+r_b=0$ per the pattern of $U'$, we conclude that $r_a=r_b$. For the sets $\s{c,d}$ and $\s{e,f}$ the corresponding rows in the residue matrix of $\oline{U}'$ are nonzero, and so by similar reasoning we conclude that $r_c\neq r_d$ and $r_e\neq r_f$. This leaves two possibilities for the pattern of $U$: either $\partition(U)$ is a $2\times2\times2$ pattern with $\s{a,b}\in\partition(U)$ and $\partition(U)\neq\partition(S')$, or $\partition(U)$ is a $2\times4$ pattern such that $\partition(U)=\s{\s{d,f},\s{a,b,c,e}}$ (up to $c\rightarrow d$ and $e\rightarrow f$), i.e that $\partition(U)\cap\s{\s{x,y}~;~(x,y)\in [3]\times[4,6]}\neq\varnothing$.
			\begin{enumerate}
				\item In the case that $U$ has a $2\times2\times2$ pattern, then the leftmost letter $S$ of $U$ is such that $\partition(S) = \partition(U) \neq \partition(S')$ and $\partition(S)\cap \partition(S') = \s{\s{a,b}}$ with $\s{a,b}\in\s{\s{x,y}~;~(x,y)\in [3]\times[4,6]}$. By inspection, we see that the only possibility is that $S\in\s{\gens_{10},\cdots,\gens_{15}}$. Noting that $\partition(S)\cap \partition(S')\neq\varnothing$, and given that
				$S'=\gens_{j'}$ and $S=\gens_j$, we conclude that
				$j\not\in\delta_{\mathfrak{S},15}(j',S'=\gens_{j'})$ As $U$ is accepted by $\mathfrak{S}_{10,15}\circ\mathfrak{C}$ by \cref{lem:222pattern} and $S'$ is accepted by $\mathfrak{S}_{1,9}$, we conclude that $U' = S' U\in\lang(\mathfrak{S}_{1,9}\circ\mathfrak{S}_{10,15}\circ\mathfrak{C}) \setminus\lang(\mathfrak{S}_{1,15}\circ\mathfrak{C})$.
				\item In the case that $U$ has a $2\times4$ pattern with $\partition(U)\cap\s{\s{x,y}~;~(x,y)\in [3]\times[4,6]}\neq\varnothing$, as the length of $U$ is $k$ we assume by the induction hypothesis that $U\in\lang(\mathfrak{S}_{1,9}\circ\mathfrak{S}_{10,15}\circ\mathfrak{C}) \setminus\lang(\mathfrak{S}_{1,15}\circ\mathfrak{C})$. Under FFP, we note that the leftmost letter $S$ of $U$ must be an element of $\s{\gens_1,\cdots,\gens_9}$ and have the pattern $\partition(S) = \s{\s{a,c},\s{b,e},\s{d,f}}$ given that $\partition(U) = \s{\s{d,f},\s{a,b,c,e}}$ (again, up to $c\rightarrow d$ and $e\rightarrow f$). Thus, $\partition(S')\cap\partition(S) = \varnothing$. Letting $S'=\gens_{j'}$ and $S=\gens_j$, we conclude that
				$j\in\delta_{\mathfrak{S},9}(j',S'=\gens_{j'})$. Because
				$S=\gens_j$ is the first letter of the word $U$, we know that the
				initial state of $U$ must be $j$. Therefore, by the
				induction hypothesis, $U' = S' U\in\lang(\mathfrak{S}_{1,9}\circ\mathfrak{S}_{10,15}\circ\mathfrak{C}) \setminus\lang(\mathfrak{S}_{1,15}\circ\mathfrak{C})$.
			\end{enumerate}
			We have exhausted all cases for $U$, and so we conclude that the leftward inclusion holds.
		\end{itemize}
	We have shown that \cref{eq:24lemrestate} holds for words of length
	$k+1$ if it holds for words of length $k$. This completes the
	inductive step.\qedhere
	\end{itemize}
\end{proof}

\begin{lemma*}
	\label{applem:24pattern3366}
	Let $U$ be a word over $\gens\cup\clifford$. Then $U\in
	\lang(\mathfrak{S}_{1,3}\circ\mathfrak{S}_{4,9}\circ\mathfrak{S}_{10,15}\circ
	\mathfrak{C})\setminus
	\lang(\mathfrak{S}_{1,9}\circ\mathfrak{S}_{10,15}\circ
	\mathfrak{C})$ if and only if $U\in\mathcal{N}$ and $U$ has a
	$2\times4$ pattern with
	$\partition(U)\cap\s{\s{x,y}~;~(x,y)\in[3]\times[4,6]}=\varnothing$.
\end{lemma*}

\begin{proof}
	By \cref{lem:clifauto,lem:222pattern,lem:24pattern36}, $\lang(\mathfrak{S}_{1,9}\circ\mathfrak{S}_{10,15}\circ\mathfrak{C})$ accounts for all Cliffords, $2\times2\times2$ pattern normal forms, and $2\times 4$ pattern normal forms where the pattern contains an element of $\s{\s{x,y}~;~(x,y)\in[3]\times[4,6]}$. Any accepting word of $\mathfrak{S}_{1,3}\circ\mathfrak{S}_{4,9}\circ\mathfrak{S}_{10,15}\circ\mathfrak{C}$ which is not accepted by $\mathfrak{S}_{1,9}\circ\mathfrak{S}_{10,15}\circ\mathfrak{C}$ is of the form $U_1 U_2 U_3$ where $U_1\in\lang(\mathfrak{S}_{1,3})\subsetneq \lang(\mathfrak{S}_{1,15})$, $U_2\in\lang(\mathfrak{S}_{4,9})\subsetneq \lang(\mathfrak{S}_{1,15})$, $U_3\in\lang(\mathfrak{S}_{10,15}\circ\mathfrak{C})\subsetneq \lang(\mathfrak{S}_{1,15}\circ\mathfrak{C})$, and $\partition(U_1^{-1})\cap\partition(U_2)\neq \varnothing$. We can then restate our lemma as follows:
	\begin{gather}
	\left\{U_1 U_2 U_3~;~ U_1\in\lang(\mathfrak{S}_{1,3}),\,U_2\in\lang(\mathfrak{S}_{4,9}),\,U_3\in\lang(\mathfrak{S}_{10,15}\circ\mathfrak{C}),\,|U_1| = \ell,\right. \nonumber\\ 
	\left. |U_2| = m,\, |U_3| = k-\ell-m,\mbox{ and } \partition(U_1^{-1})\cap \partition(U_2)\neq\varnothing\right\} \nonumber\\ 
	= \label{eq:24lemrestateOther}\\
	\s{U\in\mathcal{N}~;~ |U| = k \mbox{ and } \partition(U) \mbox { is a }2\times 4 \mbox{ pattern s.t. } \partition(U)\cap \s{\s{x,y}~;~(x,y)\in [3]\times[4,6]}=\varnothing}\nonumber
	\end{gather}
	for all $k\geq 3$, $1\leq \ell \leq k-2$, and $1\leq m\leq k-1-\ell$. We begin with a useful result, and afterwords proceed by induction on $k$.
	
	\begin{itemize}
		\item 
		Consider a length $k$ accepting word $U'$ of the above form such that $\ell=1$. Then $U' = S' U$ where, given that the first letter of $U$ is $S$, we have $S'\in\s{\gens_1,\cdots, \gens_3}$, $S\in \s{\gens_{4},\cdots, \gens_{9}}$, and $U\in\lang(\mathfrak{S}_{4,9}\circ\mathfrak{S}_{10,15}\circ\mathfrak{C})$. We know that we must have $\partition(S')\cap\partition(S)\neq\varnothing$, and by inspection the only way this is acheived is if $\partition(S')\cap \partition(S) = \s{\s{a,b}}\not\subseteq\s{\s{x,y}~;~(x,y)\in [3]\times[4,6]}$. As $\partition(S'^{-1})= \partition(S')$ is such that $\partition(S'^{-1})\cap\partition(U) = \varnothing$, we conclude that $S'^{-1}$ is not finer than $\partition(U)$ and that the least denominator exponent of $\oline{U}'$ is $k-1$ by \cref{lem:denomreduce,lem:222pattern}. In the case that $U$ has a $2\times2\times2$ pattern, we know that $\partition(U)=\partition(S)$ and so $\s{a,b}\in\partition(U)$. In the case that $U$ has a $2\times4$ pattern, we know that there exists $\s{c,d}\in\s{\s{x,y}~;~(x,y)\in [3]\times[4,6]}$ such that $\partition(S)\cap\partition(U)=\s{\s{c,d}}$. As $\s{c,d}\neq\s{a,b}$ and $\s{\s{a,b},\s{c,d}}\subseteq\partition(S)$, we know that $\s{a,b}$ is a subset of the cardinality four element of $\partition(U)$. Therefore, in both cases we have
		\[
		\rho_{k-1} (\oline{U}')[\s{a,b};[6]] = \begin{bmatrix}
		0 \\ 0
		\end{bmatrix}
		\]
		and so $U'$ must have a $2\times 4$ pattern with $\s{a,b}\in\partition(U')$. This means that
		\[
		\partition(U')\cap \s{\s{x,y}~;~(x,y)\in [3]\times[4,6]}=\varnothing
		\]
		and that $\partition(S')$ is finer than $\partition(U')$. For each such $2\times 4$ pattern, there is one and only one element of $\s{\gens_1,\cdots, \gens_3}$ which is a finer partition than $\partition(U')$, and therefore $S'$ is the leftmost syllable of the normal form equivalient to $U'$ under FFP. As $U\in\mathcal{N}$, we thus conclude that $U'\in\mathcal{N}$.
		\item We now show that accepting words of the above form with $k=3$ (enforcing $\ell=1$ and $m=1$)  constitute all such length 3 normal forms with the described pattern. Clearly, any $U\in\mathcal{N}$ with the aforementioned pattern must be of the form $U=S_1 S_2 C$ where $S_1\in\s{\gens_1,\cdots \gens_3}$, $S_2\in\gens$, and $C\in\clifford$. Furthermore, we must have $\partition(S_1)\cap \partition(S_2) = \s{\s{a,b}}$ with $\s{a,b}\not\in\s{\s{x,y}~;~(x,y)\in [3]\times[4,6]}$ to produce the appropriate $2\times 4$ pattern. By inspection, this implies that $S_2\in\s{\gens_{4},\cdots,\gens_{9}}$, and so $U$ is accepted by $\mathfrak{S}_{1,3}\circ\mathfrak{S}_{4,9}\circ\mathfrak{S}_{10,15}\circ\mathfrak{C}$ but not by $\mathfrak{S}_{1,9}\circ\mathfrak{S}_{10,15}\circ\mathfrak{C}$ as required. Therefore \cref{eq:24lemrestateOther} holds for $k=3$.
		\item Now suppose that \cref{eq:24lemrestateOther} holds for some $k\geq
		3$. We will show that
		\cref{eq:24lemrestateOther} holds for $k+1$ by establishing two
		inclusions.
		\begin{itemize}
			\item[$\subseteq$:] We have already proven this inclusion in the case of $\ell=1$ for all $k$. We therefore need only consider the $\ell>1$ case. Let $U\in\lang(\mathfrak{S}_{1,3}\circ\mathfrak{S}_{4,9}\circ\mathfrak{S}_{10,15}\circ \mathfrak{C})\setminus\lang(\mathfrak{S}_{1,9}\circ\mathfrak{S}_{10,15}\circ\mathfrak{C})$ be a word of length $k$ whose first letter is $S\in\s{\gens_1,\cdots,\gens_3}$. Then
			$U\in\mathcal{N}$ and $\partition(U)$ is a
			$2\times4$ pattern such that $\partition(U)\cap\partition(S)=\s{\s{a,b}}$ with either $a,b\in[3]$ or $a,b\in[4,6]$. By inspection, we explicitly have $S = \s{\s{a,b},\s{c,d},\s{e,f}}$ with either $\s{a,b,e}=[3]$ or $\s{a,b,e}=[4,6]$ and $\s{c,d,f}=[6]\setminus\s{a,b,e}$. Furthermore, the least denominator exponent of $\oline{U}$ is $k-1$. Suppose that $U'=S'U$ is a word of length
			$k+1$ accepted by $\mathfrak{S}_{1,3}\circ\mathfrak{S}_{4,9}\circ\mathfrak{S}_{10,15}\circ
			\mathfrak{C}$. Then by \cref{def:gensauton} we have
			$\partition(S')\cap\partition(S)=\varnothing$ with $S'\in\s{\gens_1,\cdots,\gens_3}$. By inspection, this implies that $S'$ must have the pattern $\partition(S') = \s{\s{a,e},\s{b,d},\s{c,f}}$ (up to $a\leftrightarrow b$ and $c\leftrightarrow d$) with $\s{a,b,c,d,e,f}$ as defined before. As $\partition(S'^{-1}) = \partition(S')$ is such that $\partition(S'^{-1})\cap\partition(U)=\varnothing$, then $\partition(S'^{-1})$ is not finer than $\partition(U)$ and so the least denominator exponent of $U'$ must be $k$. Consider the set $\s{c,f}$ and let $r_c$ and $r_f$ be the corresponding rows of the residue matrix of $\oline{U}$. Explicitly, we have
			\[
			\rho_{k-1}(\oline{U})[\s{c,f};[6]] = \begin{bmatrix}
			r_c \\ r_f
			\end{bmatrix}
			\]
			with $r_c = r_f$ as $\s{c,f}$ is a subset of the cardinality four element of $\partition(U)$. Directly calculating the rows of the residue matrix for $\oline{U}'$ yields
			\[
			\rho_k(\oline{U}')[\s{d,f};[6]] = \begin{bmatrix}
			r_c + r_f \\ r_c + r_f
			\end{bmatrix} = 
			\begin{bmatrix}
			0 \\ 0
			\end{bmatrix}
			\]
			and we therefore conclude that $\s{c,f}\in\partition(U')$ and $U'$ has a $2\times4$ pattern. As $\partition(S')$ is the lowest-indexed element of $\gens$ finer than $\partition(U')$, under FFP we conclude that $S'$ is the leftmost syllable of the normal form equivalent to $U'$. Since $U\in\mathcal{N}$ by assumption, we therefore conclude $U' = S' U\in\mathcal{N}$, and have established $U'$ has a $2\times4$ pattern such that $\partition(U')\cap\s{\s{x,y}~;~(x,y)\in [3]\times[4,6]}=\varnothing$.
			\item[$\supseteq$:]
			Suppose that $U'$ is a normal form of
			length $k+1$ with a $2\times 4$ pattern such that $\partition(U')\cap\s{\s{x,y}~;~(x,y)\in [3]\times[4,6]}=\varnothing$. Write $U'$
			as $U'=S'U$ for some unknown normal form $U$. We then have $S'\in\s{\gens_1,\cdots,\gens_3}$ and $\partition(S') = \s{\s{a,b},\s{c,d},\s{e,f}}$ such that either $\s{a,b,e}=[3]$ or $\s{a,b,e}=[4,6]$, $\s{c,d,f}=[6]\setminus\s{a,b,e}$, $a\equiv c\pmod3$, $b\equiv d\pmod3$, $e\equiv f\pmod3$, and $\partition(S')\cap\partition(U')=\s{\s{a,b}}$. For
			$\s{i,j}\in\partition(S')$, let the corresponding rows of
			the residue matrix of $\oline{U}$ be $r_i$ and
			$r_j$. Explicitly, we have
			\begin{align*}
			\rho_{k-1}(\oline{U})[\s{i,j};[6]]&=\begin{bmatrix}
			r_i \\
			r_j
			\end{bmatrix}.
			\end{align*}
			Direct calculation of the rows for the residue matrix of
			$\oline{U}'$ yields
			\begin{align*}
			\rho_{k}( \oline{U}')[\s{i,j};[6]]=\begin{bmatrix}
			r_i + r_j \\
			r_i + r_j
			\end{bmatrix}.
			\end{align*}
			As $r_a+r_b=0$ per the pattern of $U'$, we conclude that $r_a=r_b$. For the sets $\s{c,d}$ and $\s{e,f}$ the corresponding rows in the residue matrix of $\oline{U}'$ are nonzero, and so by similar reasoning we conclude that $r_c\neq r_d$ and $r_e\neq r_f$. This leaves three possibilites for the pattern of $U$: either $\partition(U)$ is a $2\times2\times2$ pattern with $\s{a,b}\in\partition(U)$ and $\partition(U)\neq\partition(S')$, or $\partition(U)$ is a $2\times4$ pattern such that $\partition(U)=\s{\s{c,e},\s{a,b,d,f}}$ (up to $c\leftrightarrow d$) so that $\partition(U)\cap\s{\s{x,y}~;~(x,y)\in [3]\times[4,6]}\neq\varnothing$, or $\partition(U)$ is a $2\times4$ pattern such that $\partition(U)=\s{\s{c,f},\s{a,b,d,e}}$ (up to $c\leftrightarrow d$) so that $\partition(U)\cap\s{\s{x,y}~;~(x,y)\in [3]\times[4,6]}=\varnothing$
			\begin{enumerate}
				\item In the case that $U$ has a $2\times2\times2$ pattern, then the leftmost letter $S$ of $U$ is such that $\partition(S) = \partition(U) \neq \partition(S')$ and $\partition(S)\cap \partition(S') = \s{\s{a,b}}$ with $\s{a,b}\not\in\s{\s{x,y}~;~(x,y)\in [3]\times[4,6]}$. By inspection, we see that the only possibility is that $S\in\s{\gens_{4},\cdots,\gens_{9}}$. Noting that $\partition(S)\cap \partition(S')\neq\varnothing$, and given that
				$S'=\gens_{j'}$ and $S=\gens_j$, we conclude that
				$j\not\in\delta_{\mathfrak{S},9}(j',S'=\gens_{j'})$. As $U$ is accepted by $\mathfrak{S}_{4,9}\circ\mathfrak{S}_{10,15}\circ\mathfrak{C}$ by \cref{lem:222pattern,lem:24pattern36} and $S'$ is accepted by $\mathfrak{S}_{1,3}$, we conclude that $U' = S' U\in\lang(\mathfrak{S}_{1,3}\circ\mathfrak{S}_{4,9}\circ\mathfrak{S}_{10,15}\circ\mathfrak{C}) \setminus\lang(\mathfrak{S}_{1,9}\circ\mathfrak{S}_{10,15}\circ\mathfrak{C})$.
				\item In the case that $U$ has a $2\times4$ pattern with $\partition(U)\cap\s{\s{x,y}~;~(x,y)\in [3]\times[4,6]}\neq\varnothing$, by \cref{lem:24pattern36} we know that the leftmost letter $S$ of $U$ is such that $\partition(S) \cap \partition(U) = \s{{c,e}}$ (up to $c\leftrightarrow d$). We note that $c\not\equiv d\not\equiv e\pmod3$ , and under FFP must therefore have $S\in\s{\gens_4,\cdots,\gens_9}$. This implies that we have $\partition(S) = \s{\s{a,b},\s{c,e},\s{d,f}}$ (up to $c\leftrightarrow d$), and we therefore conclude that $\partition(S')\cap\partition(S)\neq\varnothing$. Given $S'=\gens_{j'}$ and $S=\gens_j$, we conclude that
				$j\not\in\delta_{\mathfrak{S},9}(j',S'=\gens_{j'})$. As $U$ is accepted by $\mathfrak{S}_{4,9}\circ\mathfrak{S}_{10,15}\circ\mathfrak{C}$ by \cref{lem:24pattern36} and $S'$ is accepted by $\mathfrak{S}_{1,3}$, we get $U' = S' U\in\lang(\mathfrak{S}_{1,3}\circ\mathfrak{S}_{4,9}\circ\mathfrak{S}_{10,15}\circ\mathfrak{C}) \setminus\lang(\mathfrak{S}_{1,9}\circ\mathfrak{S}_{10,15}\circ\mathfrak{C})$.
				\item In the case that $U$ has a $2\times4$ pattern with $\partition(U)\cap\s{\s{x,y}~;~(x,y)\in [3]\times[4,6]}=\varnothing$, as the length of $U$ is $k$ we assume by the induction hypothesis that $U\in\lang(\mathfrak{S}_{1,3}\circ\mathfrak{S}_{4,9}\circ\mathfrak{S}_{10,15}\circ\mathfrak{C}) \setminus\lang(\mathfrak{S}_{1,9}\circ\mathfrak{S}_{10,15}\circ\mathfrak{C})$. Under FFP, we note that the leftmost letter $S$ of $U$ must be an element of $\s{\gens_1,\cdots,\gens_3}$ and have the pattern $\partition(S) = \s{\s{a,e},\s{b,d},\s{c,f}}$ given that $\partition(U) = \s{\s{c,f},\s{a,b,d,e}}$ (up to $(a,c)\leftrightarrow (b,d)$). Thus, $\partition(S')\cap\partition(S) = \varnothing$. Letting $S'=\gens_{j'}$ and $S=\gens_j$, we conclude that
				$j\in\delta_{\mathfrak{S},3}(j',S'=\gens_{j'})$. Because
				$S=\gens_j$ is the first letter of the word $U$, we know the
				initial state of $U$ must be $j$. Therefore, by the
				induction hypothesis, $U' = S' U\in\lang(\mathfrak{S}_{1,3}\circ\mathfrak{S}_{4,9}\circ\mathfrak{S}_{10,15}\circ\mathfrak{C}) \setminus\lang(\mathfrak{S}_{1,9}\circ\mathfrak{S}_{10,15}\circ\mathfrak{C})$.
			\end{enumerate}
			We have exhausted all cases for $U$, and so we conclude that the leftward inclusion holds.
		\end{itemize}
		We have shown that \cref{eq:24lemrestateOther} holds for words of length
		$k+1$ if it holds for words of length $k$. This completes the
		inductive step.\qedhere
	\end{itemize}
\end{proof}

\end{document}